\documentclass{article}
\usepackage[margin=1in]{geometry}
\usepackage[utf8]{inputenc}
\usepackage{array}
\usepackage{amsmath,amsthm}

\usepackage{graphicx,epstopdf}
\usepackage{tikz}
\usetikzlibrary{automata,positioning, matrix, decorations.pathreplacing}
\usetikzlibrary{calc, automata, chains, arrows.meta}
\usetikzlibrary{chains,shapes.multipart}
\usetikzlibrary{shapes,calc}
\usepackage{enumitem}

\def\P#1{\mathrm{\bf P}\left\{#1\right\}}	
\def\E#1{\mathrm{\bf E}\left[#1\right]}		

\newtheorem{theorem}{Theorem}
\newtheorem{definition}{Definition}

\title{Scalable Load Balancing in the Presence of Heterogeneous Servers}
\author{Kristen Gardner$^{1}$, Jazeem Abdul Jaleel$^{2}$, Alexander Wickeham$^{2}$, Sherwin Doroudi$^2$ \vspace{0.2in} \\ 
$^{1}$ Amherst College \\
$^{2}$ University of Minnesota}
\date{June 24, 2020}

\begin{document}

\maketitle

\begin{abstract}
Heterogeneity is becoming increasingly ubiquitous in modern large-scale computer systems.  Developing good load balancing policies for systems whose resources have varying speeds is crucial in achieving low response times.  Indeed, how best to dispatch jobs to servers is a classical and well-studied problem in the queueing literature.
Yet the bulk of existing work on large-scale systems assumes homogeneous servers; unfortunately, policies that perform well in the homogeneous setting can cause unacceptably poor performance---or even instability---in heterogeneous systems.

We adapt the ``power-of-$d$'' versions of both the Join-the-Idle-Queue and Join-the-Shortest-Queue policies to design two corresponding families of heterogeneity-aware dispatching policies, each of which is parameterized by a pair of routing probabilities.
Unlike their heterogeneity-unaware counterparts, our policies use server speed information both when choosing which servers to query and when probabilistically deciding where (among the queried servers) to dispatch jobs.
Both of our policy families are analytically tractable: our mean response time and queue length distribution analyses are exact as the number of servers approaches infinity, under standard assumptions.
Furthermore, our policy families achieve maximal stability and outperform well-known dispatching rules---including heterogeneity-aware policies such as Shortest-Expected-Delay---with respect to mean response time.

\end{abstract}

\section{Introduction}
\label{sec:intro}

In large-scale computer systems, deciding how to dispatch arriving jobs to servers is a primary factor affecting system performance.
Consequently, there is a wealth of literature on designing, analyzing, and evaluating the performance of load balancing policies.
For analytical tractability, most existing work on dispatching in large-scale systems makes a key assumption: that the servers are homogeneous, meaning that they all have the same speeds, capabilities, and available resources.
But this assumption is not accurate in practice.
Modern computer systems are instead heterogeneous: server farms may consist of multiple generations of hardware, servers with varied resources, or even virtual machines running in a cloud environment.
Given the ubiquity of heterogeneity in today's systems, it is critically important to develop load balancing policies that perform well in heterogeneous environments.
In this paper, we focus on systems in which \emph{server speeds} are heterogeneous.


The dominant dispatching paradigm in the contemporary literature on large scale systems is the ``power of $d$ choices,'' wherein the dispatcher cannot use global information to make dispatching decisions, as that would require prohibitively expensive computation upon each job's arrival.  Rather, a fixed number ($d$) of servers are queried at random, and a dispatching decision is made among these servers.  Unfortunately, the ``power of $d$'' policies that have been designed to perform well in homogeneous systems can lead to unacceptably poor performance---or even instability---in the presence of heterogeneity. 
For example, the classical Join-the-Shortest-Queue-$d$ (JSQ-$d$) policy, under which, upon a job's arrival, the dispatcher queries $d$ servers uniformly at random and sends the job to the queried server with the fewest jobs in its queue, can cause the system to become unstable if the system's capacity is concentrated among a relatively small number of fast servers.
JSQ-$d$ is just one example of a \emph{heterogeneity-unaware} policy, but recent work has shown that other heterogeneity-unaware policies, including Join Idle Queue (JIQ), also can lead to poor performance in heterogeneous systems.
Clearly, it is necessary to use server speed information when making dispatching decisions in heterogeneous systems.

Yet simply using heterogeneity information is not enough: it matters exactly when and how the dispatcher uses this information.
Consider the Shortest-Expected-Delay-$d$ (SED-$d$) policy, a natural heterogeneity-aware generalization of JSQ-$d$.
Under SED-$d$, upon a job's arrival the dispatcher queries $d$ servers uniformly at random and sends the job to the queried server at which the job's expected delay---the number of jobs in the queue scaled by the server's speed---is smallest.
By allowing the dispatcher to select a fast server with a longer queue over a slow server with a shorter queue, SED-$d$ overcomes one of the weaknesses of JSQ-$d$ in the presence of heterogeneity.
Unfortunately, this is insufficient to solve the fundamental problem faced by JSQ-$d$. SED-$d$, too, can cause poor performance and instability if fast servers are queried infrequently.

While server heterogeneity poses a problem for many existing dispatching policies, it also presents an opportunity to design new policies that leverage heterogeneity to achieve good performance and maintain stability, rather than suffering in the presence of heterogeneity.
Our key insight is that there are two decision points at which ``power of $d$'' policies can use server speed information.
First, the dispatcher can make heterogeneity-aware decisions about which $d$ servers to query.
Second, the dispatcher can make heterogeneity-aware decisions about where among the queried servers to send an arriving job.
Alone, neither decision point appears to be enough to both ensure stability and achieve good performance.  In combination, they allow for the design of a new class of powerful policies that benefit from server speed heterogeneity, thereby resolving the problems of instability and poor performance.

We propose two new families of policies, called JIQ-($d_F$,$d_S$) and JSQ-($d_F$,$d_S$), that are inspired by classical ``power of $d$'' policies but use server speed information at both decision points.
This enables them to significantly outperform JSQ-$d$, SED-$d$, and other heterogeneity-aware policies, as well as to maintain the full stability region.
At the first decision point, instead of quering $d$ servers uniformly at random from among all servers, our policies query $d_F$ fast servers and $d_S$ slow servers.
Unlike under JSQ-$d$ and SED-$d$, this guarantees that each job has the option to run on a fast server.
After querying $d_F + d_S$ servers, our policies decide probabilistically based on the servers' states (idle or busy) whether to dispatch the job to a fast server or a slow server.
Our policy families are analytically tractable: given the probabilistic parameter settings, we derive the mean response time and queue length distribution under each.
While the two families are functionally similar, they require different analytical approaches.
We analyze JIQ-($d_F$,$d_s$) using a mean field approach, and JSQ-($d_F$,$d_S$) using a system of differential equations capturing the system evolution.
Our analyses of both policies are exact in the limiting regime where the number of servers approaches infinity, under standard asymptotic independence assumptions.

The remainder of this paper is organized as follows. In Section~\ref{sec:relatedwork} we survey related work on dispatching in heterogeneous systems. Section~\ref{sec:model} describes the system model and defines the JIQ-($d_F,d_S$) and JSQ-($d_F,d_S$) policy families. In Section~\ref{sec:analysis} we present our analyses of both policies.
We give a numerical evaluation in Section~\ref{sec:results} and propose a heuristic for selecting policy parameters in Section~\ref{sec:heuristic}.
Finally, in Section~\ref{sec:conclusion}, we conclude.

\section{Related Work}
\label{sec:relatedwork}

In large-scale homogeneous systems, Join-the-Shortest-Queue (JSQ) is known to minimize mean response time under first-come-first-served (FCFS) scheduling when service times are independent and identically distributed and have non-decreasing hazard rate~\cite{winston1977optimality,weber1978optimal}.
While analyzing response time is challenging due to the dependencies among queue lengths, approximations exist in both the FCFS setting with exponential service times~\cite{nelson1989approximation} and the Processor Sharing (PS) setting with general service times~\cite{gupta2007analysis}.
Because of the high communication cost required to query all servers for their queue lengths, the JSQ-$d$ (also called SQ($d$) or Power-of-$d$) policy was proposed and analyzed, assuming homogeneous servers and exponential service times~\cite{mitzenmacher2001power,vvedenskaya1996queueing}.
Other policies, such as Join-Idle-Queue (JIQ), have also been proposed as low-communication alternatives to JSQ~\cite{lu2011join,wang2018distributed}.

Once the server homogeneity assumption is relaxed, the optimality and analytical tractability of state-aware dispatching policies suffers.
The SQ(2) policy has been studied in heterogeneous FCFS systems with general service times, under both light traffic~\cite{izagirre2014light} and heavy traffic~\cite{zhou2017designing} assumptions.
Performance analysis also exists for SQ(2) in heterogeneous PS systems~\cite{mukhopadhyay2016analysis}.
The Shortest Expected Delay (SED) policy is a natural alternative to JSQ when server speeds are known; SED has been shown empirically to perform favorably to several other heterogeneity-aware policies~\cite{banawan1992comparative}.
However, SED is known to be suboptimal in general~\cite{whitt1986deciding}.
When service times are generally distributed, SED requires knowledge of the full job size distribution in order to estimate the remaining service time of the job currently in service.
The Generalized JSQ (GJSQ) policy has been proposed as an alternative when only the mean job size at each server, not the full job size distribution, is known~\cite{selen2016approximate} (note that when service times are exponentially distributed, SED and GJSQ are equivalent).
The equilibrium distribution of the number of jobs in the system has been analyzed under both SED and GJSQ in a heterogeneous two-server system~\cite{selen2016approximate,selen2016steady}.
The Balanced Routing policy (which we call Weighted JSQ in Section~\ref{sec:results}) uses server speed information by querying servers probabilistically in proportion to their speeds but ignores heterogeneity information when choosing among the queried servers; this policy minimizes the system workload in heavy traffic~\cite{chen2012asymptotic}, but can be suboptimal at lower load.

A common theme in much of the recent work on dispatching in heterogeneous systems is the observation that policies like SQ($d$) and JIQ, which were designed for homogeneous systems, have a reduced stability region when used in heterogeneous systems.
Consequently, much of the recent work in heterogeneous systems has focused on developing policies that maximize the stability region.
Recently several families of throughput optimal policies have been proposed, including PULL~\cite{stolyar2015pull} and $\Pi$~\cite{zhou2017designing}.
PULL, which is similar to JIQ, is shown to be optimal in the sense that it stochastically minimizes the queue length distribution~\cite{stolyar2015pull}; as we will see in Section~\ref{sec:results}, this does not mean that it is optimal with respect to other system metrics such as response time.

Another related stream of work focuses on the so-called ``slow server problem,'' wherein the system designer must choose when to use a slow server if at all.  Typically, models consist of two servers of different speeds with all jobs arriving to a single queue \cite{larsen81,lin84,rubinovitch85,rubinovitch85_stall, koole-scl-1995}, with more recent work examining similar problems in settings with more than two servers \cite{luh2002,rykov09}.  As they examine a central queue setting rather than an immediate dispatching setting, the policies and analysis proposed in these papers are inapplicable to our setting.  Closer to our setting but still within the literature on central queues is \cite{shenker1989}, which considers dispatching to one of two subsystems: a central queue for a limited number of fast servers, and a subsystem with an infinite number of slow servers.  

More closely related to our work is a literature stream on dispatching in small-scale heterogeneous systems \cite{tantawi1985optimal, bonomi-tc-1990, banawan1989load, FENG2005, Sethuraman:1999}.  Such work explores policies that use information about all servers' queue lengths (or sometimes more detailed information, as in \cite{esa2013}) when making dispatching decisions. These are not ``power of $d$'' policies and would not typically be considered scalable; hence, our policies of interest, analytical approaches, and qualitative findings differ significantly from those in the papers above.

\section{Model}
\label{sec:model}

\begin{figure}
	\centering
	\includegraphics[scale=0.4]{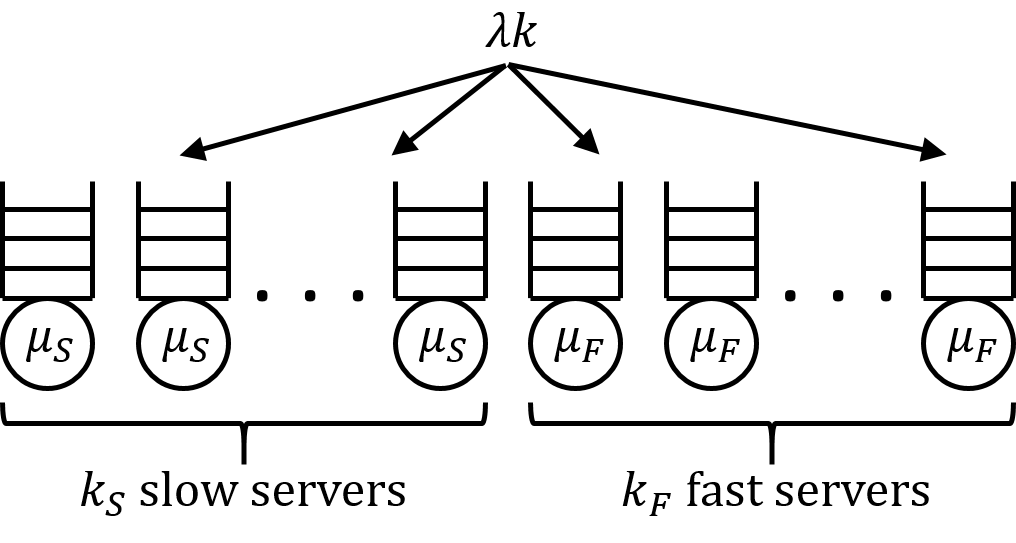}
	\caption{The system consists of $k_F$ fast servers, each with service rate $\mu_F$, and $k_S$ slow servers, each with service rate $\mu_S$. Job arrive to the system as a Poisson process with rate $\lambda k$ and are dispatched immediately.}
	\label{fig:model}
\end{figure}

Our system consists of $k$ heterogeneous servers (see Figure~\ref{fig:model}).
There are two classes of servers: $k_F$ of the servers are ``fast'' servers and $k_S = k - k_F$ of the servers are ``slow'' servers.
We let $q_F = \frac{k_F}{k}$ and $q_S = \frac{k_S}{k} = 1 - q_F$ denote the fraction of servers that are fast and slow respectively.
Service times are independent and for most of the paper we assume that they are exponentially distributed with rate $\mu_F$ on fast servers and rate $\mu_S$ on slow servers, where the speed ratio $r\equiv\mu_F/\mu_S>1$.
In Section~\ref{sec:general_service_times} we consider general service time distributions.
For simplicity, we assume that $\mu_F q_F + \mu_S q_S = 1$, so that the system has total capacity $k$.

Jobs arrive to the system as a Poisson process with rate $\lambda k$.
Upon arrival to the system, a job is dispatched immediately to a single server according to some policy.
Each server works on the jobs in its queue in first-come first-served (FCFS) order without preemption.

We consider two dispatching policies: JIQ-($d_F,d_S$) and JSQ-($d_F,d_S$).
The common framework shared by both policies favors idle fast servers whenever possible, and leverages the idea that slow servers are still occasionally worth utilizing (motivating probabilistic decision-making), and it is better to utilize them when idle rather than busy (motivating the use of two---rather than just one---probabilistic parameters).

\begin{definition}
	Under both \textbf{JIQ-($d_F,d_S$)} and \textbf{JSQ-($d_F,d_S$)}, when a job arrives the dispatcher queries $d_F$ fast servers and $d_S$ slow servers, chosen uniformly at random without replacement. The job is then dispatched to one of the queried servers as follows:
	\begin{itemize}[leftmargin=*]
		\item If any of the $d_F$ fast servers are idle, the job begins service on one of them.
		\item If all $d_F$ fast servers are busy and any of the $d_S$ slow servers are idle:
		\begin{itemize}
			\item With probability $p_S$ the job begins service on an idle slow server.
			\item With probability $1-p_S$ the job is dispatched to a \textbf{chosen} fast server among the $d_F$ queried.
		\end{itemize}
		\item If all $d_F + d_S$ queried servers are busy:
		\begin{itemize}
			\item With probability $p_F$ the job is dispatched to a \textbf{chosen} fast server among the $d_F$ queried.
			\item With probability $1-p_F$ the job is dispatched to a \textbf{chosen} slow server among the $d_S$ queried.
		\end{itemize}
	\end{itemize}
	The difference between the two policies lies in how a busy server (among those under consideration) is \textbf{chosen}.
	Under JIQ-($d_F$,$d_S$) the server is chosen uniformly at random. 
	Under JSQ-($d_F$,$d_S$) the server with the shortest queue is chosen. Under both policies all ties are broken uniformly at random.  
\end{definition}

\section{Analysis}
\label{sec:analysis}

In this section we analyze the queue length distribution and mean response time under both JIQ-$(d_F,d_S$) and JSQ-($d_F,d_S$).
Let $\rho_F$ and $\rho_S$ denote respectively the fraction of time that a fast server is busy and that a slow server is busy.
We begin with the observation that $\rho_F$ and $\rho_S$ are independent of the choice of policy between JIQ-($d_F,d_S$) and JSQ-($d_F,d_S$) and of the service time distribution.
For both policies, and for any service time distribution such that the system is stable, we have
\begin{align}
\label{eq:rhof}
\rho_F = \lambda k \P{\mathrm{job ~runs ~on ~a ~fast ~server}} \cdot \frac{1}{\mu_F k_F} 
&= \frac{\lambda}{\mu_F q_F} \left( (1-\rho_F^{d_F}) + \rho_F^{d_F}(1- \rho_S^{d_S})(1-p_S) + \rho_F^{d_F} \rho_S^{d_S} p_F \right) \\
\rho_S = \lambda k \P{\mathrm{job ~runs ~on ~a ~slow ~server}} \cdot \frac{1}{\mu_S k_S} 
&= \frac{\lambda}{\mu_S q_S} \left( \rho_F^{d_F} (1-\rho_S^{d_S})p_S + \rho_F^{d_F} \rho_S^{d_S}(1-p_F) \right).
\label{eq:rhos}
\end{align}
Solving this system of equations, numerically if an exact analytical solution is not possible, yields $\rho_F$ and $\rho_S$.
We will define $\pi_{0F} = 1 - \rho_F$  (respectively, $\pi_{0S} = 1 - \rho_S$) to be the probability that a fast (slow) server is idle.

We will assume that $k \rightarrow \infty$ and that in this limiting regime the queue lengths at each of the servers become independent.
This lets us treat a single queue as its own isolated system.
While we do not formally prove this asymptotic independence, our numerical results indicate that as $k$ becomes large our approximation is highly accurate.

\subsection{JIQ-($d_F$,$d_S$)}
\label{sec:jiq-analysis}

We will derive performance metrics under JIQ-($d_F,d_S$) first for exponential service times, then for general service times.
For both analyses, we use a mean field approach and study a tagged fast server and a tagged slow server, each in isolation. 
We will need the arrival rates to fast and slow servers when they are busy and when they are idle; we note that these rates are independent of the service time distribution.
Let $\lambda_{BF}$, $\lambda_{IF}$, $\lambda_{BS}$, and $\lambda_{IS}$ denote respectively the arrival rates to a tagged busy fast, idle fast, busy slow, and idle slow server.

Let $\lambda_{QF}$ denote the arrival rate of jobs that query a tagged fast server.
We have
\begin{equation}
\label{eq:lambdaQF}
\lambda_{QF} = \lambda k \frac{\binom{k_F - 1}{d_F - 1}\binom{k_S}{d_S}}{\binom{k_F}{d_F}\binom{k_S}{d_S}} = \frac{\lambda d_F}{q_F}.
\end{equation}
$\lambda_{QS}$ is defined similarly.  The arrival rates $\lambda_{IF}$ and $\lambda_{BF}$ depend not only on the state of the tagged fast server, but also on whether the \emph{other} servers queried by an arriving job are busy or idle.
Under our asymptotic independence assumption, all other fast (respectively, slow) servers have the same stationary distribution, $\pi_{iF}$ ($\pi_{iS}$), as the tagged fast (slow) server, where $\pi_{iF}$ ($\pi_{iS}$) denotes the stationary probability that there are $i$ jobs at a tagged fast (slow) server, $i\in\{0,1,\ldots\}$.
When the tagged fast server is idle, an arriving job that queries the tagged server will be dispatched to it if it is chosen (uniformly at random) among all idle fast servers queried by the arrival.
We have:
\begin{equation}
\label{eq:lambdaIF}
\lambda_{IF} = \lambda_{QF} \left( \sum_{i=0}^{d_F-1} \binom{d_F-1}{i} \frac{\pi_{0F}^{i} (1 - \pi_{0F})^{d_F-1-i}}{i+1}  \right).
\end{equation}
When the tagged fast server is busy, an arriving job that queries the tagged server will be dispatched to it if none of the other queried fast servers are idle (probability $(1-\pi_{0F})^{d_F-1}$), and if either (1) the arrival queries an idle slow server (probability $1-\left(1-\pi_{0S}\right)^{d_S}$), the dispatcher chooses to send the job to a fast server (probability $1-p_S$), and the tagged fast server is chosen uniformly at random among all queried fast servers (probability $1/d_F$), or (2) all queried slow servers are busy (probability $\left(1-\pi_{0S}\right)^{d_S}$), the dispatcher chooses to send the job to a fast server (probability $p_F$), and the tagged fast server is chosen uniformly at random among all queried fast servers (probability $1/d_F$).
We thus have:
{\small
	\begin{equation}
	\label{eq:lambdaBF}
	\lambda_{BF} = \lambda_{QF} \frac{(1-\pi_{0F})^{d_F-1}}{d_F}  \left( \left( 1 - (1 - \pi_{0S})^{d_S} \right) (1-p_S) + (1 - \pi_{0S})^{d_S} p_F  \right). 
	\end{equation}
}

Our approach for the tagged slow server is similar, yielding:
\begin{align}
\label{eq:lambdaIS}
\lambda_{IS} &= \lambda_{QS} (1-\pi_{0F})^{d_F} \left( \sum_{i=0}^{d_S-1} \binom{d_S - 1}{i} \frac{\pi_{0S}^{i}(1 - \pi_{0S})^{d_S-1-i}}{i+1} p_S \right) \\
\label{eq:lambdaBS}
\lambda_{BS} &= \lambda_{QS} (1-\pi_{0F})^{d_F} \frac{(1 - \pi_{0S})^{d_S-1}}{d_S} (1-p_F).
\end{align}

We are now ready to derive mean response time under both exponential and general service times.

\subsubsection{Exponential service times}

Our approach involves setting up and solving a Markov chain for a tagged fast server and for a tagged slow server.
We begin with the fast server.
Recall that state $iF$ denotes that there are $i$ jobs at the fast server, including the job in service if there is one, and $\pi_{iF}$ denotes that state's stationary probability. 
The number of jobs at the tagged fast server will evolve as a state-dependent M/M/1 queue with arrival rate $\lambda_{IF}$ when it is idle, arrival rate $\lambda_{BF}$ when it is busy, and service rate $\mu_F$.
Figure~\ref{fig:tagged_server} depicts the Markov chain corresponding to this server.
\begin{figure}[t]
	\centering
	\begin{tikzpicture}[start chain = going right,
	-Triangle, every loop/.append style = {-Triangle}]
	\foreach \i in {0,...,3} 
	\node[state, on chain]  (\i) {\i F};
	\node[draw = none, on chain]  (dots) {$\dots$};
	\draw 
	(0) edge[bend left, auto = left] node {$\lambda_{IF}$} (1)
	(1) edge[bend left, auto = left] node {$\mu_F$} (0)
	(3) edge[bend left, auto = left, pos = 0.45] node {$\lambda_{BF}$} (dots)
	(dots) edge[bend left, auto = left, pos = 0.55] node {$\mu_F$} (3)
	;
	\foreach \i in {1,...,2} {
		\draw let \n1 = { int(\i+1) } in
		(\i)  edge[bend left, auto = left] node {$\lambda_{BF}$} (\n1)
		(\n1) edge[bend left, auto = left] node {$\mu_F$} (\i);
	}
	\end{tikzpicture} 
	\caption{The Markov chain tracking the number of jobs at a tagged fast server.  State $iF$ indicates that there are $i$ jobs at the fast server (including the job in service, if there is one).}
	\label{fig:tagged_server}
\end{figure}
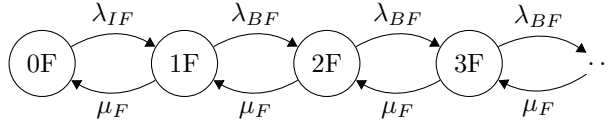

The stationary probabilities for this Markov chain are:
\begin{equation*}
\label{eq:balance_eqn_fast}
\pi_{iF} = \frac{\lambda_{IF}}{\mu_F} \left( \frac{\lambda_{BF}}{\mu_F} \right)^{i-1} \pi_{0F}, \qquad i \geq 1.
\end{equation*}
With the normalization equation, $\sum_{i=0}^{\infty} \pi_{iF} = 1$, this yields:
\begin{equation}
\label{eq:pi0F}
\pi_{0F} = \frac{\mu_F - \lambda_{BF}}{\mu_F - \lambda_{BF} + \lambda_{IF}}.
\end{equation}

Our approach for the slow server is similar, yielding:
\begin{align}
\pi_{iS} &= \frac{\lambda_{IS}}{\mu_S} \left( \frac{\lambda_{BS}}{\mu_S} \right)^{i-1} \pi_{0S}, \qquad i \geq 1. \nonumber \\
\pi_{0S} &= \frac{\mu_S - \lambda_{BS}}{\mu_S - \lambda_{BS} + \lambda_{IS}} \label{eq:pi0S} 
\end{align}

We now have six equations (\ref{eq:lambdaIF},\ref{eq:lambdaBF},\ref{eq:lambdaIS},\ref{eq:lambdaBS},\ref{eq:pi0F},\ref{eq:pi0S}) to solve for six unknown variables ($\pi_{0F}$, $\lambda_{IF}$, $\lambda_{BF}$, $\pi_{0S}$, $\lambda_{IS}$, $\lambda_{BS}$), after which we will have obtained the full queue length distribution under JIQ-($d_F,d_S$).

We are now ready to give an expression for mean response time as a function of the system parameters and the policy parameters $p_F$ and $p_S$.
Let $\E{N_F}$ and $\E{N_S}$ denote respectively the mean number of jobs at a fast server and at a slow server.
We have:
\begin{align}
\E{N_F} &= \sum_{i=0}^{\infty} i \pi_{iF} = \pi_{0F} \frac{\lambda_{IF}}{\mu_F} \sum_{i=1}^{\infty} i \left( \frac{\lambda_{BF}}{\mu_F} \right)^{i-1} \nonumber \\
&= \frac{\lambda_{IF}\mu_F}{(\mu_F - \lambda_{BF})(\mu_F - \lambda_{BF} + \lambda_{IF})} \\
\E{N_S} &= \frac{\lambda_{IS}\mu_S}{(\mu_S - \lambda_{BS})(\mu_S - \lambda_{BS} + \lambda_{IS})}.
\end{align}
Putting this together, the mean number of jobs in the system is:
\begin{equation}
\E{N} = k_F \E{N_F} + k_S \E{N_S}.
\end{equation}
Finally, we apply Little's Law to obtain the mean response time:
\begin{align}
\label{eq:etjiq}
\E{T} &= \frac{k_F \E{N_F} + k_S \E{N_S}}{\lambda k} \nonumber \\
&= \frac{q_F \lambda_{IF}\mu_F}{\lambda (\mu_F - \lambda_{BF})(\mu_F - \lambda_{BF} + \lambda_{IF})} + \frac{q_S \lambda_{IS}\mu_S}{\lambda (\mu_S - \lambda_{BS})(\mu_S - \lambda_{BS} + \lambda_{IS})}.
\end{align}

\subsubsection{General service times}
\label{sec:general_service_times}

For general service times, our Markov chain approach no longer applies.
Now, a job's service time on a fast server (respectively, a slow server) is distributed like $Y_F$ ($Y_S$), where $ Y_S \sim r Y_F$.  Note that the servers exhibit heterogeneity in speed, but (as in the case of exponential service times) the coefficient of variation associated with service times is the same across both server speeds.

To analyze this system, we make the observation that the dynamics of a busy fast server are identical to those of an M/G/1 system with arrival rate $\lambda_{BF}$ and service time distributed like $Y_F$.
The only difference between these two systems is that they have different arrival rates when idle; this does not affect the response time distribution.
Hence we can conclude that the response time distribution at a fast server under JIQ-($d_F,d_S$) is the same as that of this  M/G/1.
A similar result holds for slow servers.
The Pollaczek-Khinchine formula gives us:
\begin{align*}
\E{T_F} &= \frac{\lambda_{BF} \E{Y_F^2} }{2(1 - \lambda_{BF} \E{Y_F})} + \E{Y_F} \\
\E{T_S} &= \frac{\lambda_{BS} \E{Y_S^2} }{2(1 - \lambda_{BS} \E{Y_S})} + \E{Y_S}.
\end{align*}

Conditioning on whether an arriving job is dispatched to a fast or a slow server, we then obtain the system mean response time:
\begin{align}
\label{eq:etjiq_gen}
\E{T} &= \frac{q_F \lambda_F}{\lambda} \left( \frac{\lambda_{BF} \E{Y_F^2} }{2(1 - \lambda_{BF} \E{Y_F})} + \E{Y_F} \right) + \frac{q_S \lambda_S}{\lambda} \left( \frac{\lambda_{BS} \E{Y_S^2} }{2(1 - \lambda_{BS} \E{Y_S})} + \E{Y_S} \right),
\end{align}
which coincides with (\ref{eq:etjiq}) when $Y_F$ and $Y_S$ are exponentially distributed.

The observation that a tagged fast server essentially behaves like an M/G/1 also allows us to adapt standard techniques, such as M/G/1 transform analysis, to derive queue length distributions and other system metrics (see Chapter 26 of \cite{Mor-2013}).

\subsubsection{Optimization}
\label{sec:jiq-opt}

Having determined $\E{T}$ for a fixed $p_F$ and $p_S$, we can now optimize the JIQ-($d_F$,$d_S$) policy by finding the optimal values for $p_F$ and $p_S$.
We will assume a fixed $d_F$ and $d_S$, but note that we could also optimize over $d_F$ and $d_S$; only a small set of values for $d_F$ and $d_S$ are likely to be practical.

Equation (\ref{eq:etjiq_gen}) tells us that mean response time is linear in the second moments of $Y_F$ and $Y_S$.
This means that, because $Y_F$ and $Y_S$ have the same coefficient of variation, the optimal values of $p_F$ and $p_S$ depend only on the mean service times $\E{Y_F} = 1/\mu_F$ and $\E{Y_S} = 1/\mu_S$.
This insensitivity property allows us to assume exponential service times without loss of generality when carrying out our optimization.

Our optimization problem is as follows:
\begin{equation}
\begin{aligned}
\label{opt:jiq}
& \underset{p_F, p_S}{\text{minimize}} 
& & \E{T}\\
& \text{subject to} 
& & \mathrm{Equations} (\ref{eq:lambdaIF},\ref{eq:lambdaBF},\ref{eq:lambdaIS},\ref{eq:lambdaBS},\ref{eq:pi0F},\ref{eq:pi0S})\\
& & & 0 \leq \pi_{0F}, \pi_{0S} \leq 1 \\
& & & 0 \leq p_F, p_S \leq 1
\end{aligned}
\end{equation}
where $\E{T}$ is given in (\ref{eq:etjiq}).  We provide an explicit formulation of this problem in the Appendix.

\subsection{JSQ-($d_F$,$d_S$)}
\label{sec:jsq-analysis}

While the difference between JIQ-($d_F$,$d_S$) and JSQ-($d_F$,$d_S$) may seem like only a minor policy modification, it necessitates a fundamentally different analytical approach.
Imagine applying the tagged server approach used to analyze JIQ-($d_F$,$d_S$) to JSQ-($d_F$,$d_S$), and consider a tagged fast server under JSQ-($d_F$,$d_S$).
As under JIQ-($d_F$,$d_S$), this server experiences a state-dependent arrival rate.
Unlike under JIQ-($d_F$,$d_S$), this arrival rate is different for \emph{every} state, and it depends on the queue lengths of all other polled servers.
Hence adopting the Markov chain-based approach we used for JIQ-($d_F$,$d_S$) would require solving a highly complicated infinite system of equations.

Instead, our approach for analyzing JSQ-($d_F$,$d_S$) will involve considering a tagged \emph{arrival} to the system, again assuming that $k \rightarrow \infty$ and that in this limiting regime, all servers have independent queue lengths.

We condition on whether the tagged arrival runs on a fast or slow server and on whether or not it waits in the queue:
\begin{align}
\E{T} &= \E{T|\mathrm{run ~on ~idle ~fast}} \cdot \P{\mathrm{run ~on ~idle ~fast}} + \E{T|\mathrm{run ~on ~idle ~slow}} \cdot \P{\mathrm{run ~on ~idle ~slow}} \nonumber \\
&\quad + \E{T|\mathrm{queue ~at ~busy ~fast}} \cdot \P{\mathrm{queue ~at ~busy ~fast}} + \E{T|\mathrm{queue ~at ~busy ~slow}} \cdot \P{\mathrm{queue ~at ~busy ~slow}} \nonumber \\
&= \frac{1}{\mu_F} \cdot (1-\rho_F^{d_F}) + \frac{1}{\mu_S} \cdot \rho_F^{d_F}(1-\rho_S^{d_s})p_S + \E{T|\mathrm{queue ~at ~busy ~fast}} \cdot \rho_F^{d_F}(\rho_S^{d_S} p_F + (1-\rho_S^{d_s})(1-p_S)) \nonumber \\
&\quad + \E{T|\mathrm{queue ~at ~busy ~slow}} \cdot \rho_F^{d_F} \rho_S^{d_S}(1-p_F). \label{eq:et}
\end{align}
In line (\ref{eq:et}) we use the asymptotic independence assumption.

We next derive $\E{T|\mathrm{~queue ~at ~busy ~fast}}$.
Here, the job joins the shortest queue among the $d_F$ polled fast servers, all of which are busy.
In order to derive response time, we first need to determine the distribution of the number of jobs in a fast server's queue.

Let $n_i(t)$ denote the number of fast servers with \emph{at least} $i$ jobs at time $t$.
Let $f_i(t) = n_i(t)/k_F$ be the fraction of servers that are fast and have at least $i$ jobs at time $t$.
We note that $f_0(t) = 1$ for all $t$.

As in~\cite{mitzenmacher2001power}, we consider a limiting system, where $k \rightarrow \infty$ and the system exhibits deterministic steady-state behavior where $df_i(t)/dt=0$ for all $i\ge0$.  This setting lets us describe our system's evolution through a system of differential equations wherein all $f_i(t)$ functions are constant (henceforth we write $f_i$ rather than $f_i(t)$).

We formulate the differential equations by considering the expected change in the number of fast servers' queues with at least $i > 1$ jobs over a small interval of time $dt$.
This number will increase if an arriving job joins the queue at a fast server with exactly $i-1$ jobs.
The rate at which jobs arrive to the overall system is $\lambda k$; with probability $f_{i-1}^{d_F} - f_i^{d_F}$ all $d_F$ of the polled fast servers have at least $i-1$ jobs, but not all $d_F$ have at least $i$ jobs (that is, the shortest queue among the $d_F$ fast servers contains exactly $i$ jobs).
The arriving job will join the length-($i-1$) queue if either (1) there is an idle slow server among the $d_S$ polled slow servers (probability $1-\rho_S^{d_S}$) and the job is assigned to join the queue at a fast server (probability $1-p_S$), or (2) there are no idle slow servers among the $d_S$ polled slow servers (probability $\rho_S^{d_S}$) and the job is assigned to join the queue at a fast server (probability $p_F$).
The number of queues with at least $i > 1$ jobs will decrease if a job departs from a queue with exactly $i$ jobs.
This happens with rate $\mu_F k_F (f_i - f_{i+1})$.
Putting this together, we have, for $i > 1$:
\begin{equation*}
\label{eq:dni}
\frac{d n_i}{dt} = \lambda k \left(f_{i-1}^{d_F} - f_i^{d_F}\right) \left( (1 - \rho_S^{d_S})(1 - p_S) + \rho_S^{d_S} p_F \right) - \mu_F k_F (f_i - f_{i+1}).
\end{equation*}
The case where $i = 1$ is similar, except here an arriving job that finds a fast server with $i-1 = 0$ jobs in the queue will simply begin service on that server with probability 1.
So for $i = 1$ we have:
\begin{equation*}
\label{eq:dn1}
\frac{d n_1}{dt} = \lambda k \left(f_0^{d_F} - f_1^{d_F}\right) - \mu_F k_F (f_1 - f_2).
\end{equation*}

Dividing by $k_F$ gives us a system of equations for the $f_i$ terms:
\begin{align}
\frac{d f_i}{dt} &= \frac{\lambda}{q_F} \left(f_{i-1}^{d_F} - f_i^{d_F}\right) \left( (1 - \rho_S^{d_S})(1 - p_S) + \rho_S^{d_S} p_F \right) - \mu_F (f_i - f_{i+1}). \label{eq:dfi} \\
\frac{d f_1}{dt} &= \frac{\lambda}{q_F} \left(1 - f_1^{d_F}\right) - \mu_F (f_1 - f_2). \label{eq:df1}
\end{align}
recalling that $q_F = \frac{k_F}{k}$ is the fraction of servers that are fast and that $f_0 = 1$. 
We further note that $f_1$ is the fraction of servers that are busy; using our asymptotic independence assumption, we have $f_1 = \rho_F$.
We now set $\frac{d f_i}{dt} = 0$ for all $i$ and solve for the $f_i$ terms.

Once we have the $f_i$ terms, we can find $\E{T|\mathrm{queue ~at ~busy ~fast}}$ by conditioning on the queue length seen by an arriving job:
\begin{align}
\E{T\left|\substack{\mathrm{queue ~at} \\ \mathrm{busy ~fast}}\right.} &= \sum_{i=1}^{\infty} \P{\left. \substack{\mathrm{job ~joins ~queue} \\ \mathrm{with~} i \mathrm{~jobs}}\right| \substack{\mathrm{queue ~at} \\ \mathrm{busy ~fast}}} (i+1) \frac{1}{\mu_F} \nonumber \\
&= \frac{1}{\mu_F} \sum_{i=1}^{\infty} (i+1) \cdot \frac{f_i^{d_F} - f_{i+1}^{d_F}}{f_1^{d_F}}. \label{eq:etqf}
\end{align}
Note that the probability that a job joins a queue with $i$ jobs is \emph{not} the same as the probability that a server has $i$ jobs in its queue.

Our approach to find $\E{T|\mathrm{queue ~at ~busy ~slow}}$ is similar.
Let $s_i(t)$ denote the fraction of slow servers with at least $i$ jobs at time $t$ (we will write $s_i$ when the meaning is clear).
We obtain the following system of differential equations for the $s_i$ terms:
\begin{align}
\frac{d s_i}{dt} &= \frac{\lambda}{q_S} \left(s_{i-1}^{d_S} - s_i^{d_S}\right) \rho_F^{d_F} (1 - p_F) - \mu_S (s_i - s_{i+1}). \label{eq:dsi} \\
\frac{d s_1}{dt} &= \frac{\lambda}{q_S} \left(1 - s_1^{d_S}\right) \rho_F^{d_F} p_S - \mu_S (s_1 - s_2), \label{eq:ds1}
\end{align}
where we note that $s_0(t) = 1$ for all $t$.
Again, setting $\frac{d s_i}{dt} = 0$ for all $i$ allows us to solve for a fixed point for the $s_i$ terms.

As with the fast servers, we now find
\begin{align}
\E{T\left|\substack{\mathrm{queue ~at} \\ \mathrm{busy ~slow}}\right.} 
&= \frac{1}{\mu_S} \sum_{i=1}^{\infty} (i+1) \cdot \frac{s_i^{d_S} - s_{i+1}^{d_S}}{s_1^{d_S}}. \label{eq:etqs}
\end{align}

The overall system mean response time results from combining (\ref{eq:et}), (\ref{eq:etqf}), and (\ref{eq:etqs}).

\subsubsection{Optimization}
\label{sec:jsq-opt}

As under JIQ-($d_F$,$d_S$), we now find the values of $p_F$ and $p_S$ that minimize mean response time under JSQ-($d_F$,$d_S$) (assuming $d_F$ and $d_S$ are fixed).
Our optimization problem is as follows:
\begin{equation}
\begin{aligned}
\label{opt:jsq}
& \underset{p_F, p_S}{\text{minimize}} 
& & \E{T}\\
& \text{subject to} 
& & \mathrm{Equations} (\ref{eq:rhof},\ref{eq:rhos})\\
& & & \frac{df_i}{dt} = \frac{ds_i}{dt} = 0 & i \ge 1\\
& & & f_0 = s_0 = 1 \\
& & & f_1 = \rho_F \\
& & & s_1 = \rho_S \\
& & & 0 \leq \rho_F, \rho_S \leq 1 \\
& & & 0 \leq p_F, p_S \leq 1
\end{aligned}
\end{equation}
where $\E{T}$ is given in (\ref{eq:et}, \ref{eq:etqf}, \ref{eq:etqs}) and $\frac{df_i}{dt}$, $\frac{ds_i}{dt}$ are given in (\ref{eq:dfi},\ref{eq:df1}, \ref{eq:dsi}, \ref{eq:ds1}).  We provide an explicit formulation of this problem in the Appendix.

\subsection{Stability}

One of the significant downsides to heterogeneity-unaware dispatching policies such as JSQ-$d$ and SED-$d$ is that they can become unstable under certain system parameters, including, for example, when $q_F$ is low and the fast servers are significantly faster than the slow servers.
In Theorem~\ref{thm:stability}, we show that JIQ-($d_F,d_S$) and JSQ-($d_F,d_S$) do not suffer this downside: instead, our policies remain stable as long as $\lambda < 1$, thereby achieving the maximum possible stability region.

\begin{theorem}
	\label{thm:stability}
	Under both JIQ-($d_F$,$d_S$) and JSQ-($d_F$,$d_S$) with optimal choices of $p_F$ and $p_S$, the system is stable for $\lambda < \mu_F q_F + \mu_S q_S = 1$, for any values of $d_F,d_S \geq 1$.
\end{theorem}
\begin{proof}
	We will begin by showing that the system is stable under JIQ-($d_F$,$d_S$) when $p_S = 1$ and $p_F = \mu_F q_F$, for all $d_F, d_S \geq 1$.
	The system's stability is affected by the arrival rates to busy fast servers and to busy slow servers.
	The arrival rate to an individual busy fast server, denoted $\lambda_{BF}$ (while we use the same notation as earlier in the section, note that here we do not assume that $k \rightarrow \infty$), is:
	\begin{align*}
	\lambda_{BF} &= \lambda k \frac{\binom{k_F - 1}{d_F - 1}}{\binom{k_F}{d_F}} \P{\substack{\mathrm{all ~other ~queried} \\ \mathrm{fast ~servers ~busy}}} \frac{1}{d_F} \cdot \left( \P{\substack{\mathrm{all ~queried} \\ \mathrm{slow ~servers ~busy}}} p_F + \P{\substack{\mathrm{not ~all ~queried} \\ \mathrm{slow ~servers ~busy}}} (1-p_S)  \right),
	\end{align*}
	which is at most $\lambda p_F/q_F$
	because $p_S = 1$, $\P{\mathrm{all ~other ~fast ~servers ~busy}}\leq 1$, and $\P{\mathrm{all ~queried ~slow ~servers ~busy}} \leq 1$.
	Let $p_F = \mu_F q_F$.
	Then we have $\lambda_{BF}\leq \frac{\lambda}{q_F} p_F = \lambda \mu_F<\mu_F$, ensuring the stability of the fast servers, if $\lambda < 1$.
	
	We also must consider the arrival rate to a busy slow server, denoted $\lambda_{BS}$.
	We have:
	\begin{align*}
	\lambda_{BS} &= \lambda k \frac{\binom{k_S - 1}{d_S - 1}}{\binom{k_S}{d_S}} \P{\substack{\mathrm{all ~other ~queried} \\ \mathrm{slow ~servers ~busy}}} \P{\substack{\mathrm{all ~queried} \\ \mathrm{fast ~servers ~busy}}} (1-p_F) \frac{1}{d_S},
	\end{align*}
	which is at most $\lambda (1-p_F)/q_S$ because $\P{\mathrm{all ~other ~slow ~servers ~busy}} \leq 1$ and $\P{\mathrm{all ~queried ~fast ~servers ~busy}} \leq 1$.
	Again, let $p_F = \mu_F q_F$. Then we have
	\begin{align*}
	\lambda_{BS} &\leq \frac{\lambda}{q_S} (1-p_F) = \frac{\lambda}{q_S} \mu_S q_S = \lambda \mu_S,
	\end{align*}
	which is less than $\mu_S$, ensuring the stability of the slow servers, if $\lambda < 1$.
	
	At this point we have shown that JIQ-($d_F$,$d_S$) is stable for $p_S = 1$, $p_F = \mu_F q_F$.
	We obtain the same stability result for JSQ-($d_F$,$d_S$) by observing that joining the shortest queue among $d_F$ fast servers (or among $d_S$ slow servers) instead of routing randomly to one of those $d_F$ fast servers ($d_S$ slow servers) cannot change the stability region.
	Finally, optimizing over all possible choices of $p_F$ and $p_S$ cannot decrease the stability region.
\end{proof}

Theorem~\ref{thm:stability} tells us that there always exist settings for $p_S$ and $p_F$ for which the system is stable; in Theorem~\ref{thm:stability_highlambda} we identify more specific necessary and sufficient conditions for stability as $\lambda \rightarrow 1$.

\begin{theorem}
	\label{thm:stability_highlambda}
	As $\lambda \rightarrow 1$, the system is unstable if $p_F \neq \mu_F q_F$, and the system is stable if $p_F = \mu_F q_F$ and $p_S \geq \mu_S q_S$.
\end{theorem}

\begin{proof}
	We first show that the system is stable if $p_F = \mu_F q_F$ and $p_S \geq \mu_S q_S$.
	We begin by considering an arbitrary tagged fast server.
	Note that the arrival rate to the tagged server when it is idle does not affect the stability region of that server. The arrival rate to a tagged busy fast server is
	\begin{align}
	\label{lambdaBF}
	\lambda_{BF} &= \lambda k \frac{\binom{k_F - 1}{d_F - 1}}{\binom{k_F}{d_F}} \P{\substack{\mathrm{all ~other ~queried} \\ \mathrm{fast ~servers ~busy}}} \frac{1}{d_F} \cdot \left( \P{\substack{\mathrm{all ~queried} \\ \mathrm{slow ~servers ~busy}}} p_F + \P{\substack{\mathrm{not ~all ~queried} \\ \mathrm{slow ~servers ~busy}}} (1-p_S)  \right)
	\end{align}
	We have $\P{\mathrm{all ~other ~queried ~fast ~servers ~busy}} \leq 1$, $p_F = \mu_F q_F$, and $p_S \geq \mu_S q_S$, so $1 - p_S \leq 1 - \mu_S q_S = \mu_F q_F$. Applying these bounds to (\ref{lambdaBF}) we obtain
	\begin{align*}
	\lambda_{BF} &\leq \frac{\lambda}{q_F} \left( \P{\substack{\mathrm{all ~queried} \\ \mathrm{slow ~servers ~busy}}} \mu_F q_F + \P{\substack{\mathrm{not ~all ~queried} \\ \mathrm{slow ~servers ~busy}}} \mu_F q_F  \right) = \lambda \mu_F,
	\end{align*}
	which is less than $\mu_F$, ensuring the stability of the tagged server---and hence, of all fast servers---if $\lambda < 1$.

	We now establish the stability of the slow servers.
	Because the fast servers are stable as $\lambda \rightarrow 1$, it must also be the case that $\P{\mathrm{tagged ~fast ~server ~busy}} \rightarrow 1$. 
	Thus an arriving job is likely to query $d_F$ busy servers: $\P{\mathrm{all ~queried ~fast ~servers ~busy}} \rightarrow 1$.
	Let $\P{\mathrm{all ~queried ~fast ~servers ~busy}} = 1 - \epsilon$ for some small $\epsilon > 0$, where $\epsilon \rightarrow 0$ as $\lambda \rightarrow 1$.
	The total arrival rate to all slow servers is then $\lambda k (1-\epsilon)$. 
	Consider an arbitrary tagged slow server, and note that, as for the fast servers, the arrival rate to a slow server when it is idle does not affect its stability region. For a tagged busy slow server, we have
	\begin{align*}
	\lambda_{BS} = \frac{\lambda d_S}{q_S} (1-\epsilon) \cdot \P{\substack{\mathrm{all ~other ~queried} \\ \mathrm{~slow ~servers ~busy}}} \cdot (1-p_F) \cdot \frac{1}{d_S}.
	\end{align*}
	We have $\P{\mathrm{all ~other ~queried ~slow ~servers ~busy}} \leq 1$ and $p_F = \mu_F q_F$, so $1 - p_F = 1 - \mu_F q_F = \mu_S q_S$, which gives
	\begin{align*}
	\lambda_{BS} \leq \lambda \mu_S (1-\epsilon).
	\end{align*}
	This is less than $\mu_S$, ensuring stability of the tagged slow server---and hence, of all slow servers---if $\lambda < 1$.
	
	We now turn to the second part of the result: that the system is unstable when $p_F \neq \mu_F q_F$ (for any choice of $p_S$). 
	The argument hinges on the observation that the maximum throughput of the system is $k(\mu_F q_F + \mu_S q_S) = k$ (because $\mu_F q_F + \mu_S q_S = 1$).
	In order for the system to be stable as $\lambda \rightarrow 1$ and the total system arrival rate approaches $k$, it must therefore be the case that the probability that all servers are busy approaches 1; if some servers were idle with probability $\epsilon > 0$, then the maximum possible system throughput would be less than the arrival rate and the system would be unstable.
	
	With this observation in mind, we first consider the case where $p_F > \mu_F q_F$.
	Recall from Theorem~\ref{thm:stability} the arrival rate to an individual busy fast server:
	\begin{align*}
	\lambda_{BF} &= \lambda k \frac{\binom{k_F - 1}{d_F - 1}}{\binom{k_F}{d_F}} \P{\substack{\mathrm{all ~other ~queried} \\ \mathrm{fast ~servers ~busy}}} \frac{1}{d_F} \cdot \left( \P{\substack{\mathrm{all ~queried} \\ \mathrm{slow ~servers ~busy}}} p_F + \P{\substack{\mathrm{not ~all ~queried} \\ \mathrm{slow ~servers ~busy}}} (1-p_S)  \right) \\
	&= \frac{\lambda}{q_F} \P{\substack{\mathrm{all ~other ~queried} \\ \mathrm{fast ~servers ~busy}}}  \cdot \left( \P{\substack{\mathrm{all ~queried} \\ \mathrm{slow ~servers ~busy}}} p_F + \P{\substack{\mathrm{not ~all ~queried} \\ \mathrm{slow ~servers ~busy}}} (1-p_S)  \right).
	\end{align*}
	Assuming that $\P{\substack{\mathrm{all ~other ~queried} \\ \mathrm{fast ~servers ~busy}}} \rightarrow 1$ and $\P{\substack{\mathrm{all ~queried} \\ \mathrm{slow ~servers ~busy}}} \rightarrow 1$ (if not, the system already is unstable), as $\lambda \rightarrow 1$ we have that
	$\lambda_{BF} \rightarrow p_F/q_F,$
	which is less than $\mu_F$ if $p_F < \mu_F q_F$; this contradicts our assumption that $p_F > \mu_F q_F$, hence the system is unstable in this case.
	The case where $p_F < \mu_F q_F$ is similar.
\end{proof}

It is possible that the system also remains stable for a wider range of values for $p_S > 0$, but identifying the full stability region remains an open problem.

\section{Numerical Results}
\label{sec:results}

In this section we present a numerical study to evaluate performance under the JIQ-($d_F$,$d_S$) and JSQ-($d_F$,$d_S$) policy families.
For each set of system parameters considered, we report results for the optimal policy within each family, i.e., $p_F$ and $p_S$ are chosen to minimize mean response time, as discussed in Sections~\ref{sec:jiq-opt} and~\ref{sec:jsq-opt}.
We consider different levels of server heterogeneity by varying two parameters: $q_F$ (the fraction of servers that are fast) and $r\equiv\mu_F / \mu_S$ (the speed ratio).
Unless otherwise specified, we set $d_F = d_S = 2$.

\subsection{Convergence in $k$}

\begin{figure}
	\centering
	\includegraphics[scale=0.49]{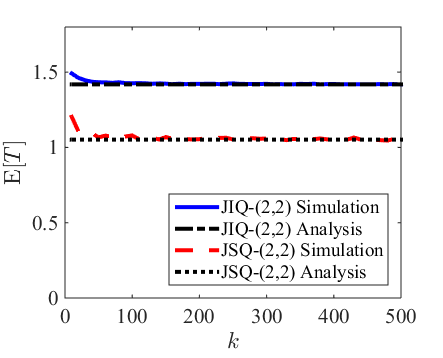}
	\caption{Analytical and simulated mean response time as a function of $k$ under both JIQ-($d_F$,$d_S$) and JSQ-($d_F$,$d_S$). Here $q_F=0.5$, $r=10$, $d_F=d_S=2$, and $p_F$ and $p_S$ are optimized separately for each policy family.}
	\label{fig:convergence}
\end{figure}

Our analyses for both JIQ-($d_F$,$d_S$) (Section~\ref{sec:jiq-analysis}) and JSQ-($d_F$,$d_S$) (Section~\ref{sec:jsq-analysis}) are approximate because they assume that the server states are independent as the number of servers $k \rightarrow \infty$.
We evaluate the accuracy of our approximations by comparing our analytical results to simulation (see Figure~\ref{fig:convergence}).
As $k$ increases our analytical results for mean response time under both policies become increasingly accurate.
By $k=500$, the analytical and simulation results are indistinguishable.
We obtained similar results for other system parameter settings.

\subsection{Mean Response Time}
\label{sec:jiq_jsq_results}

\begin{figure}[h!]
	\small{
		\begin{tabular}{>{\centering\arraybackslash}m{1.4cm}>{\centering\arraybackslash}m{4.5cm}>{\centering\arraybackslash}m{4.5cm}>{\centering\arraybackslash}m{4.5cm}}
			& {\large $q_F$ = 0.2} & {\large $q_F$ = 0.5} & {\large $q_F$ = 0.8}\\
			{\large $r$ = 1.1} &
			\includegraphics[scale=0.4]{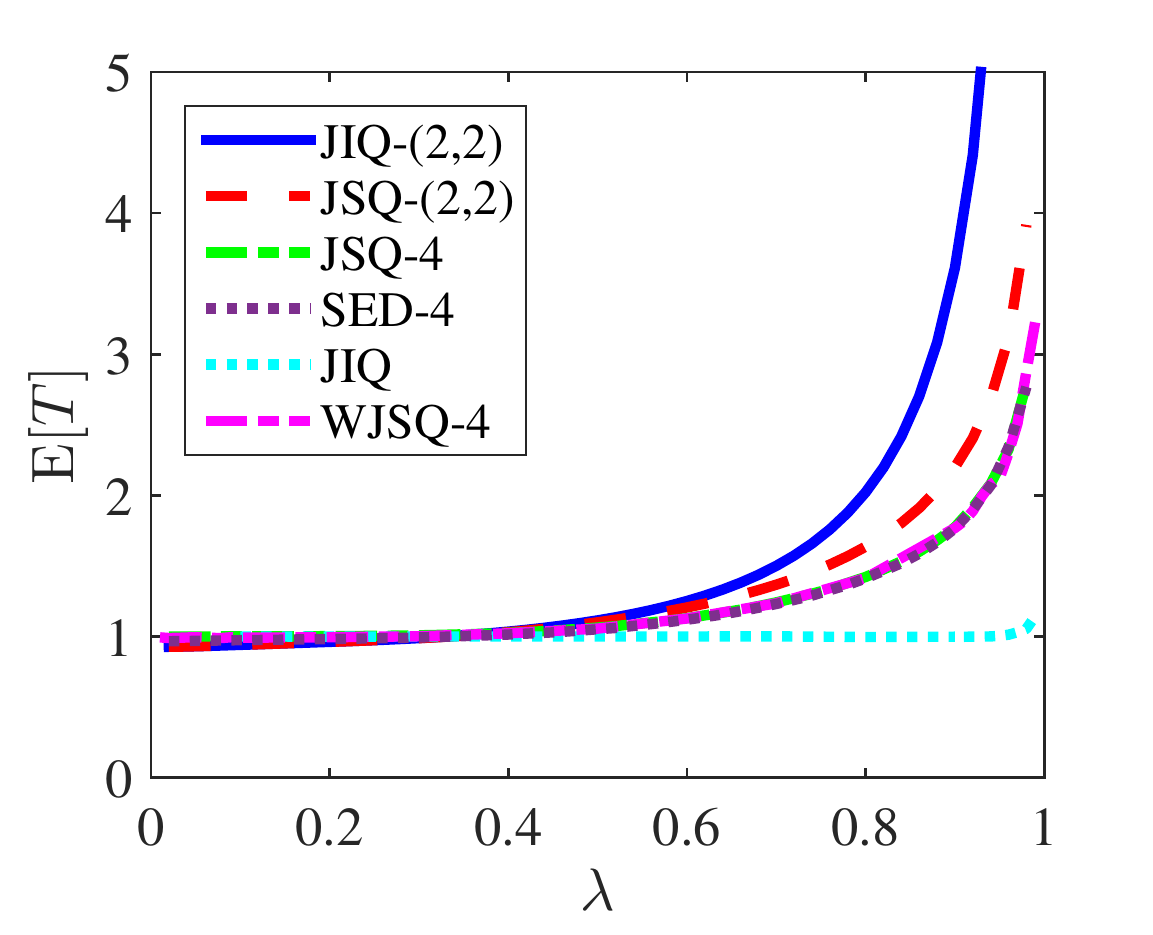} &
			\includegraphics[scale=0.4]{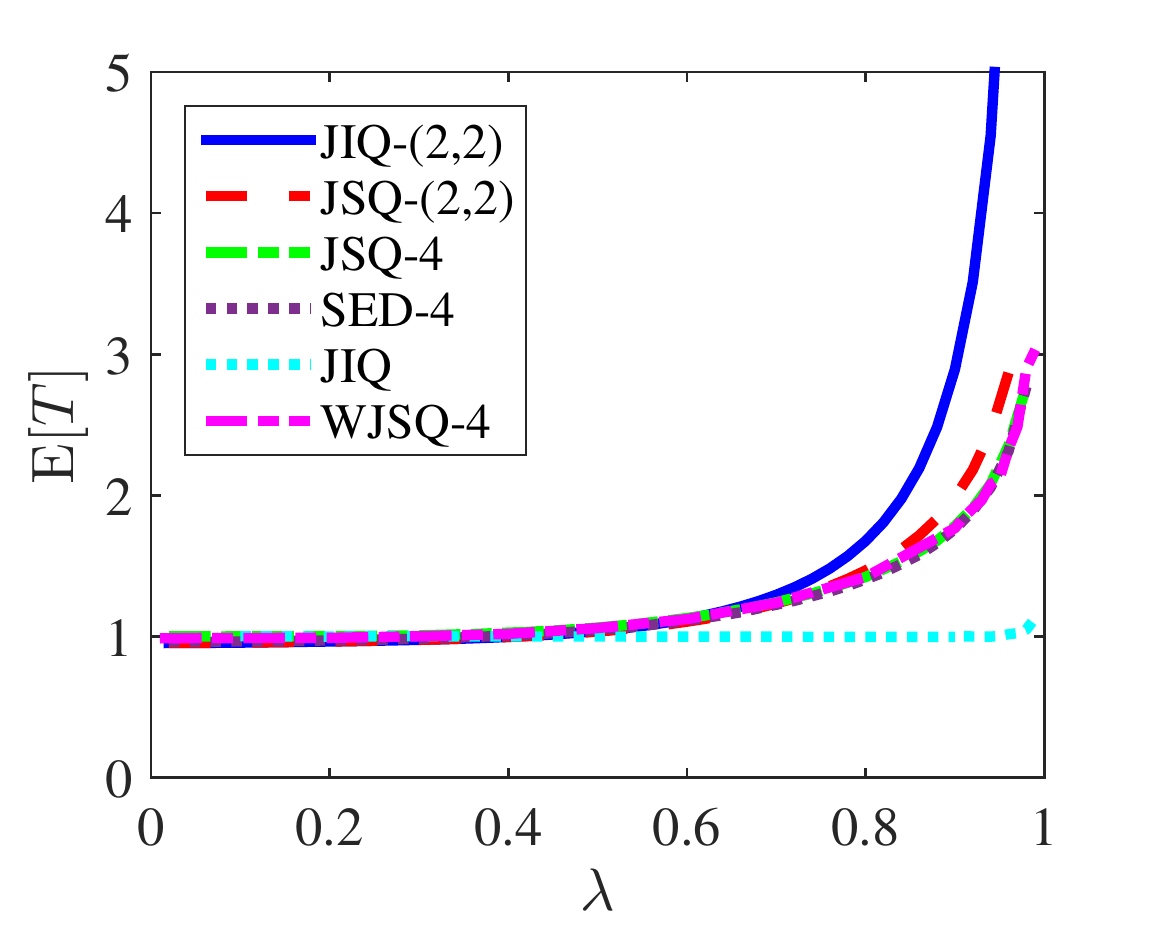} &
			\includegraphics[scale=0.4]{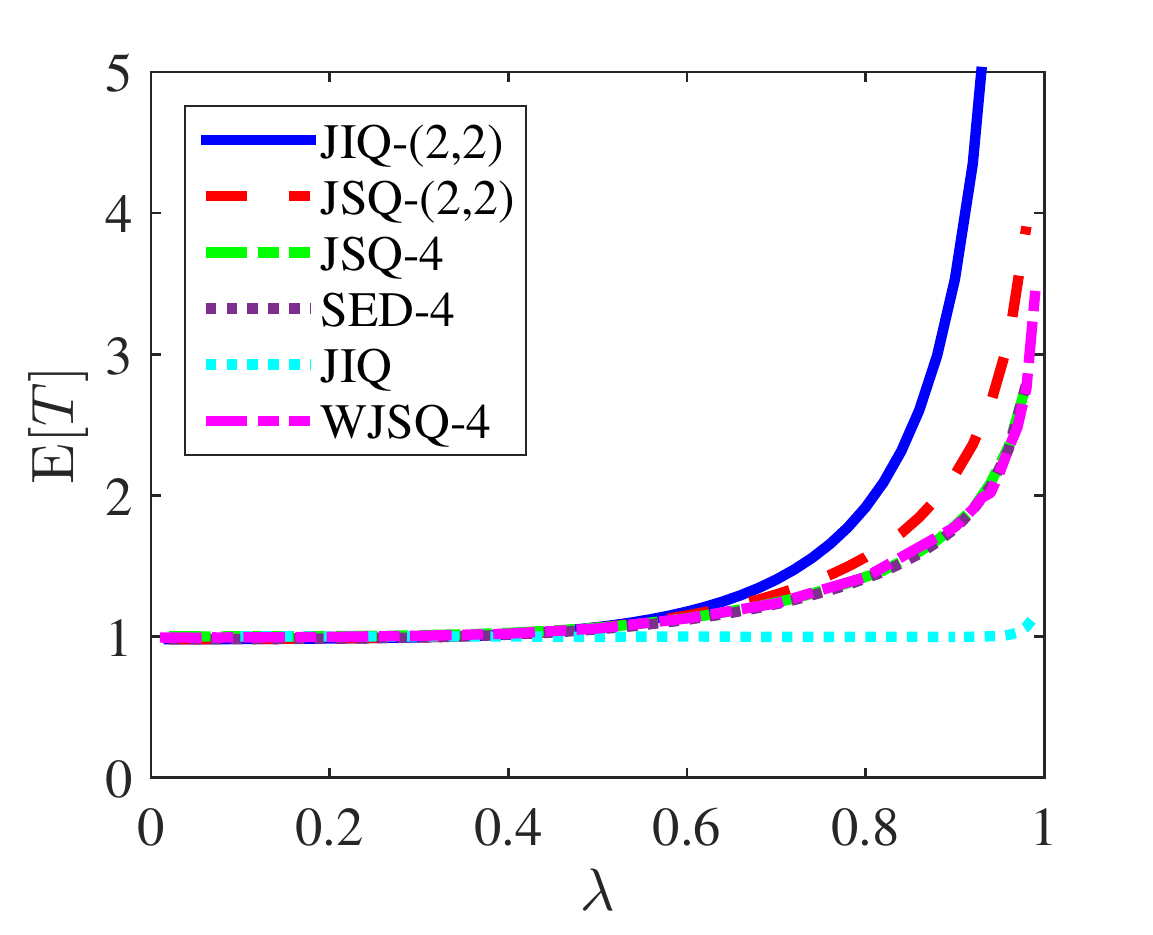} \\
			{\large $r$ = 2} &
			\includegraphics[scale=0.4]{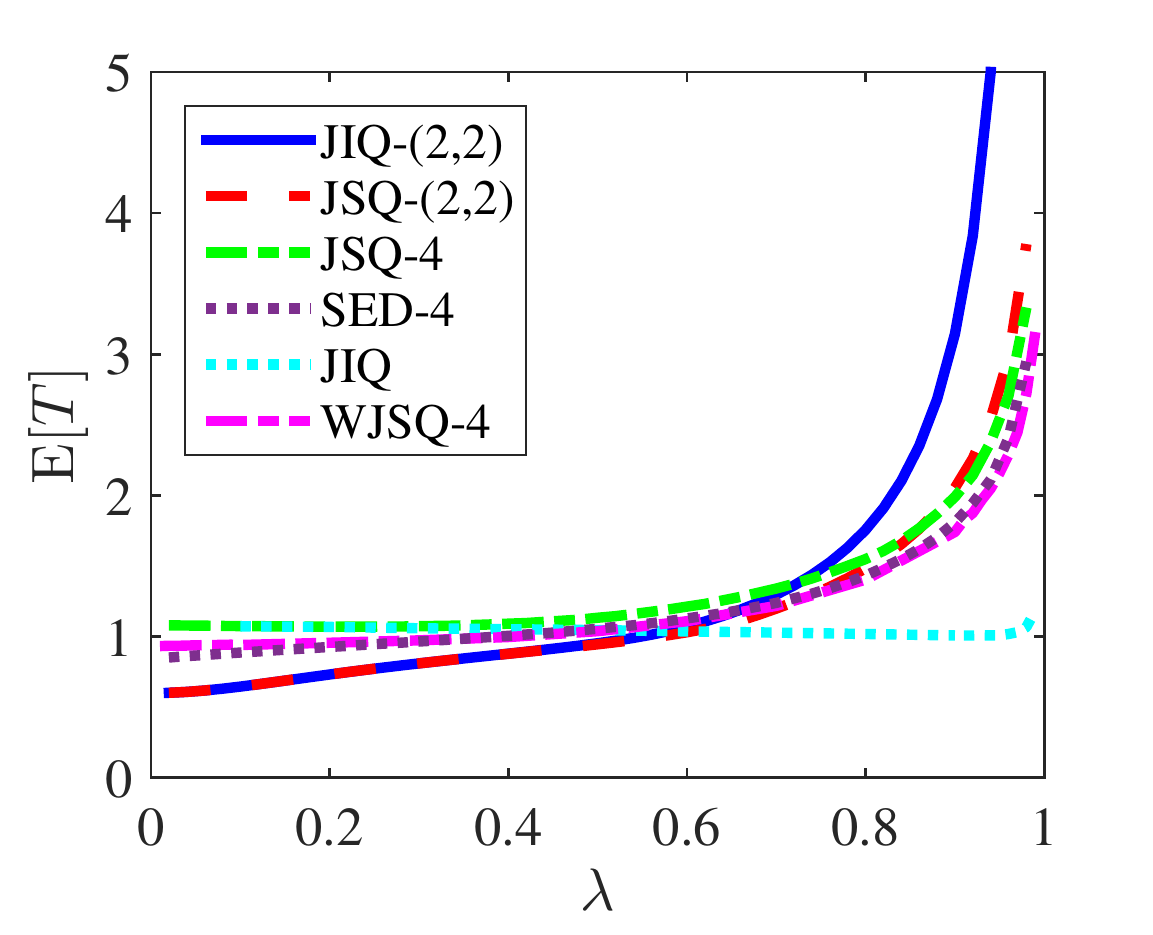} &
			\includegraphics[scale=0.4]{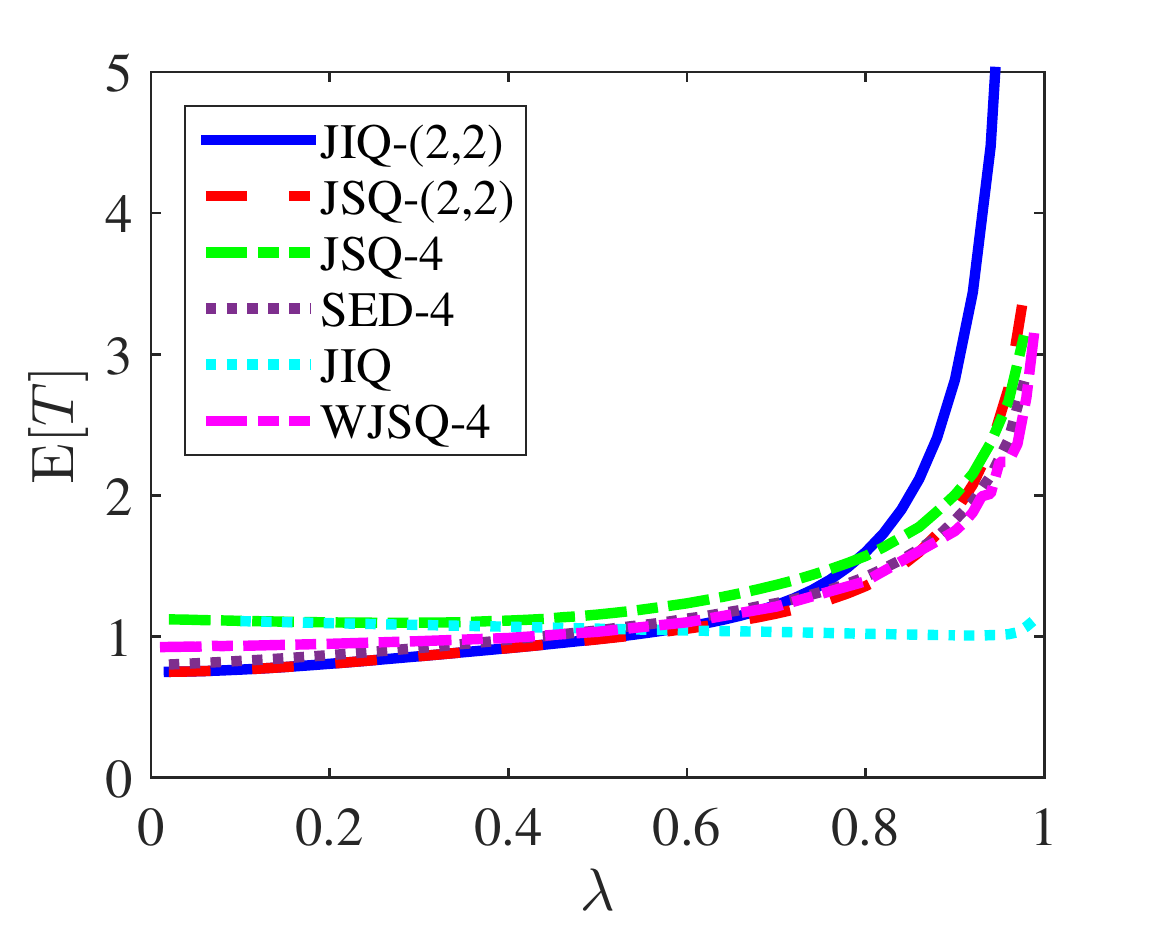} &
			\includegraphics[scale=0.4]{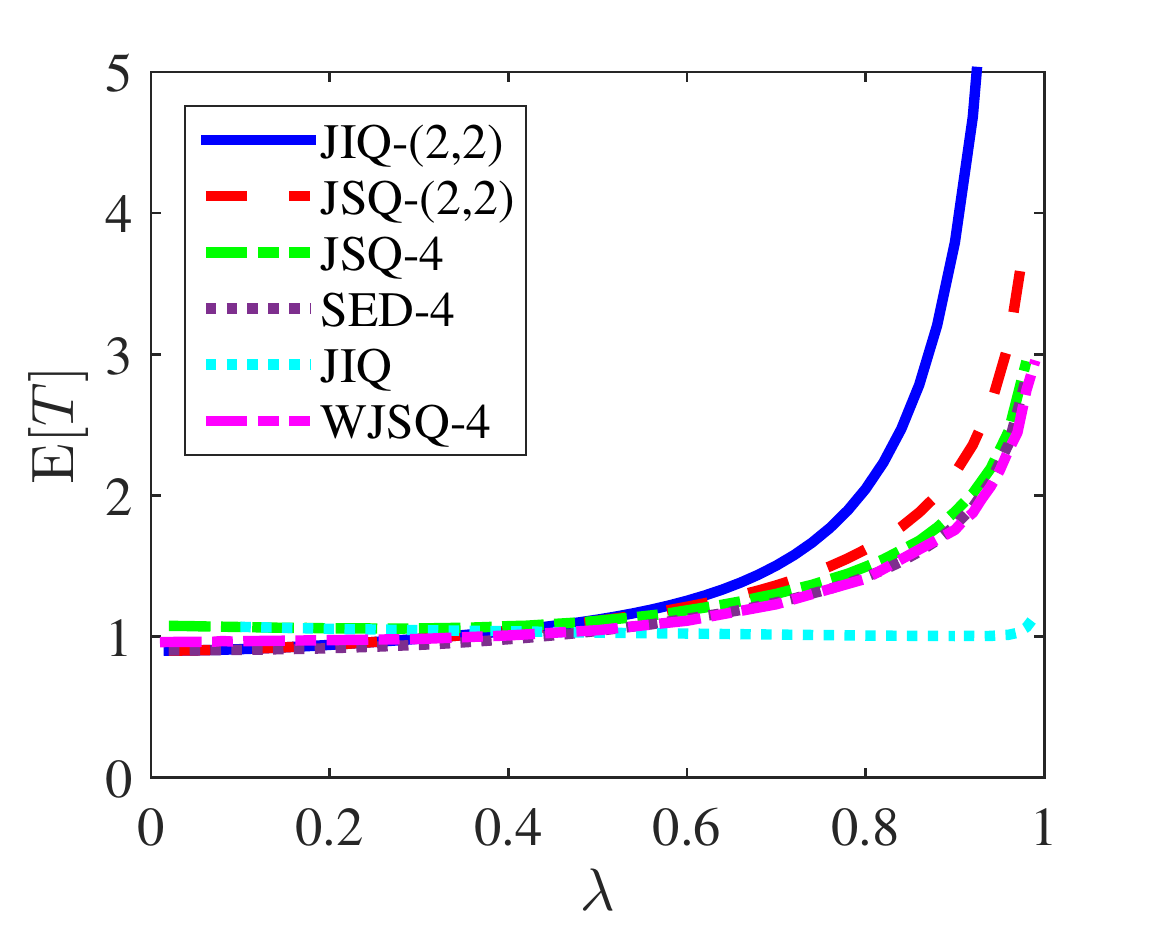} \\
			{\large $r$ = 5} &
			\includegraphics[scale=0.4]{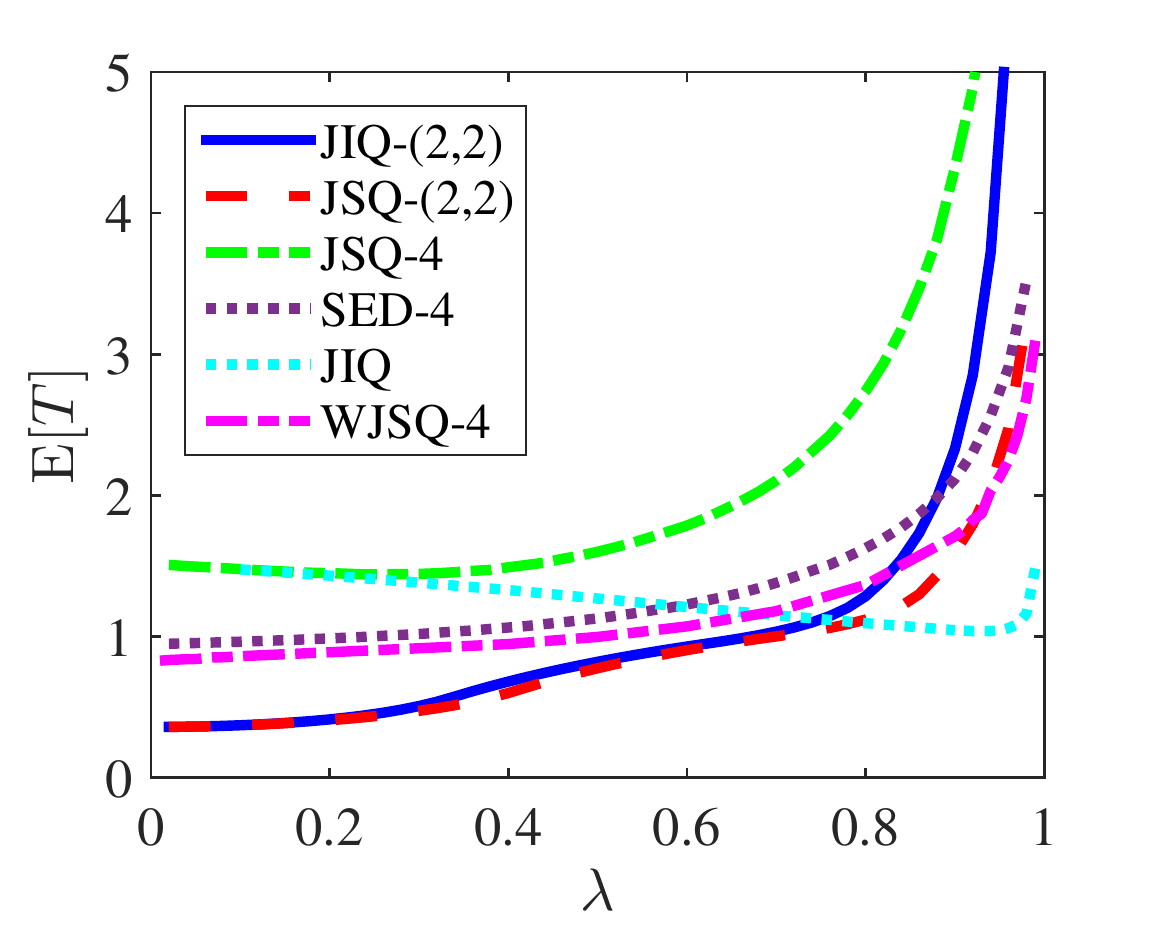} &
			\includegraphics[scale=0.4]{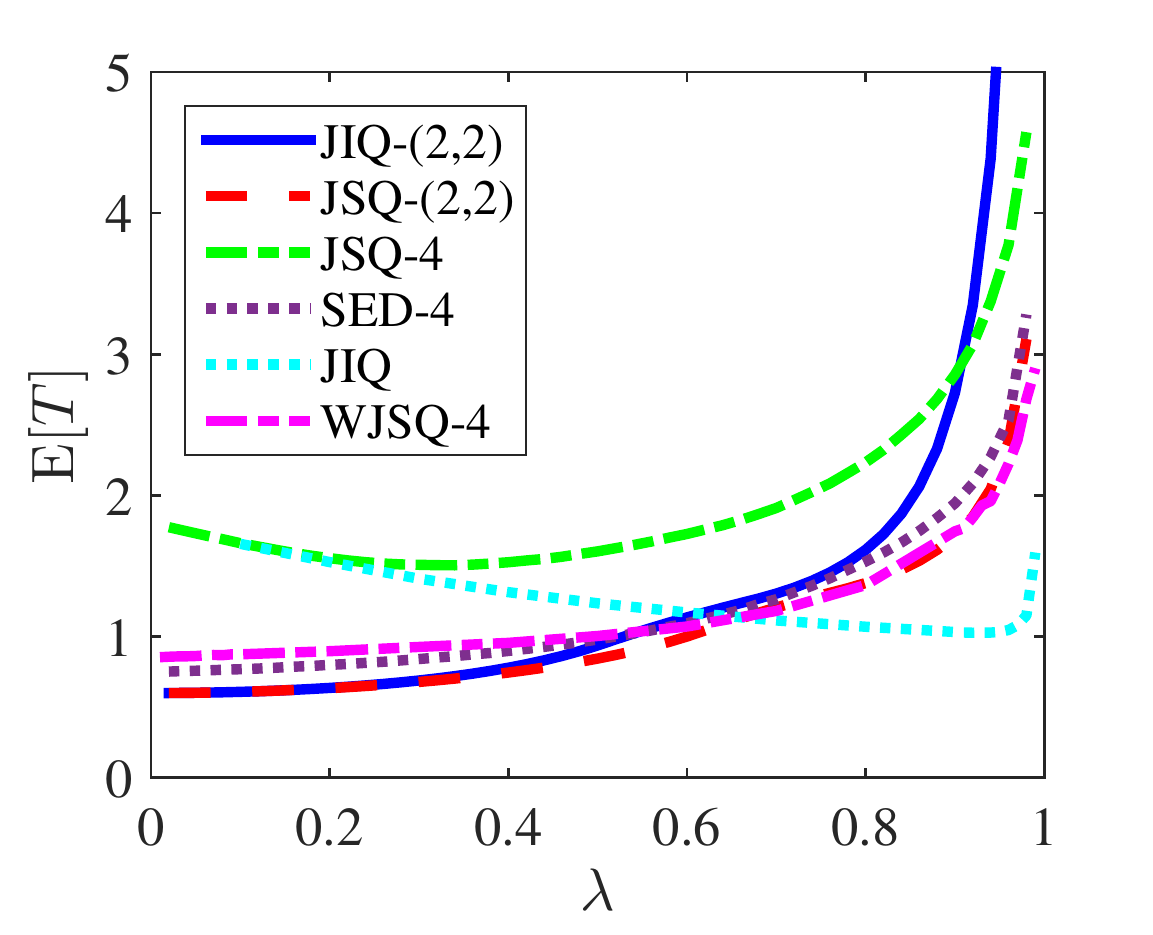} &
			\includegraphics[scale=0.4]{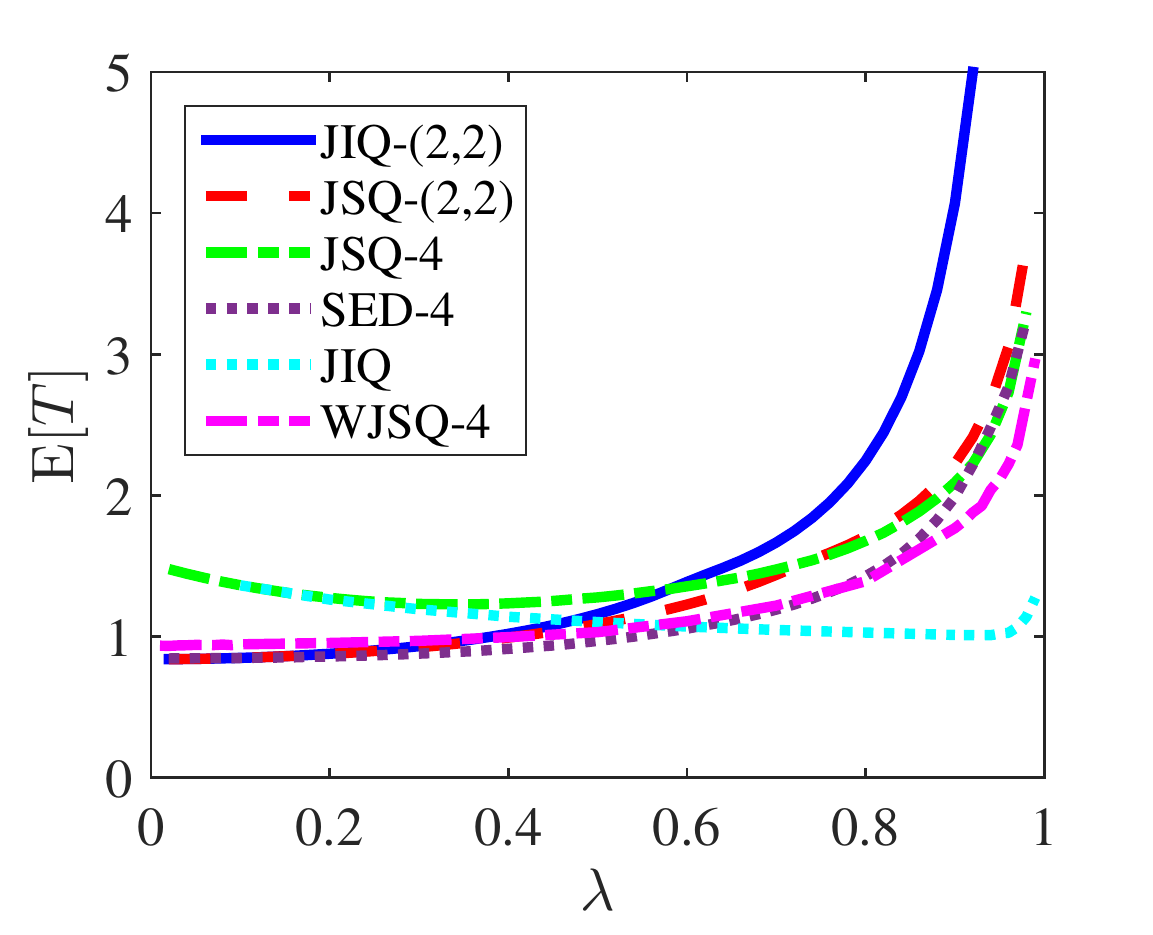} \\
			{\large $r$ = 10} &
			\includegraphics[scale=0.4]{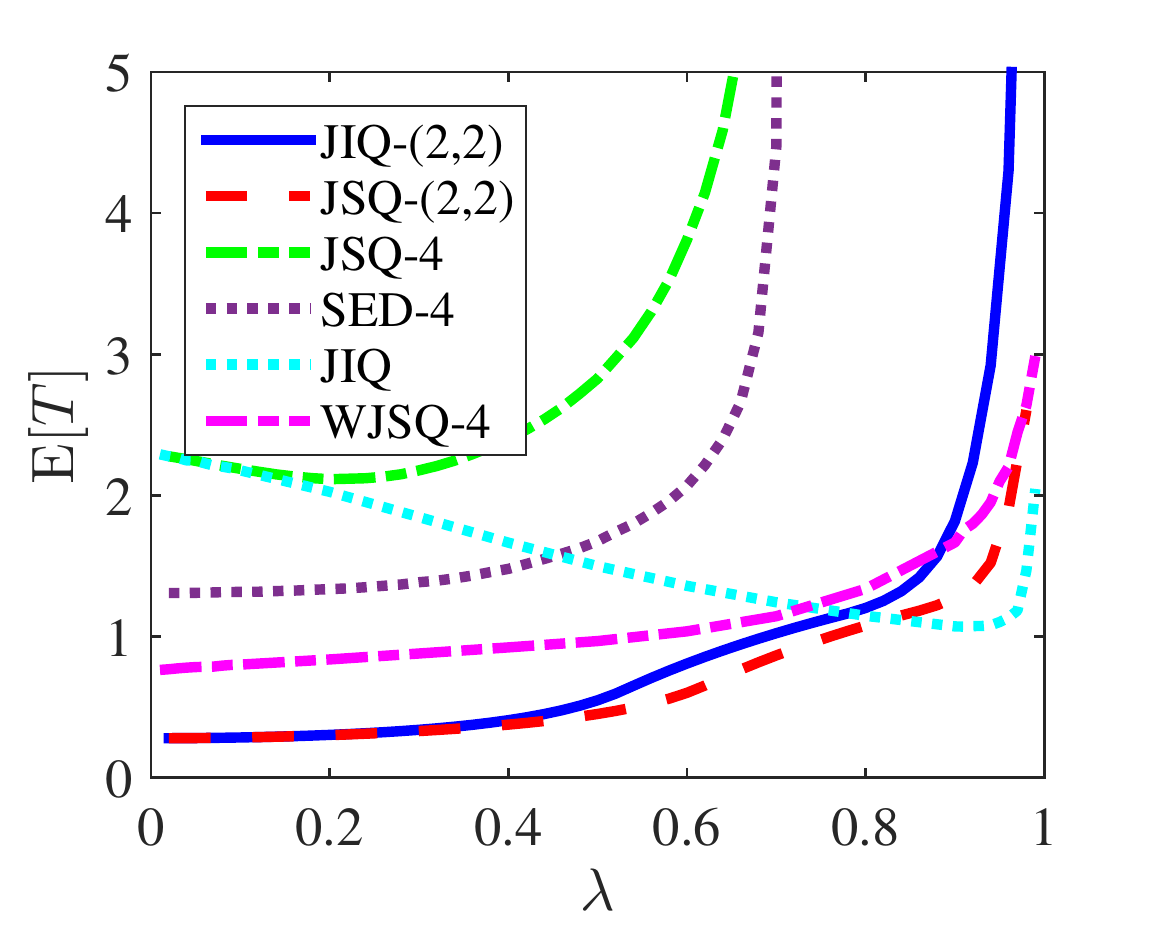} &
			\includegraphics[scale=0.4]{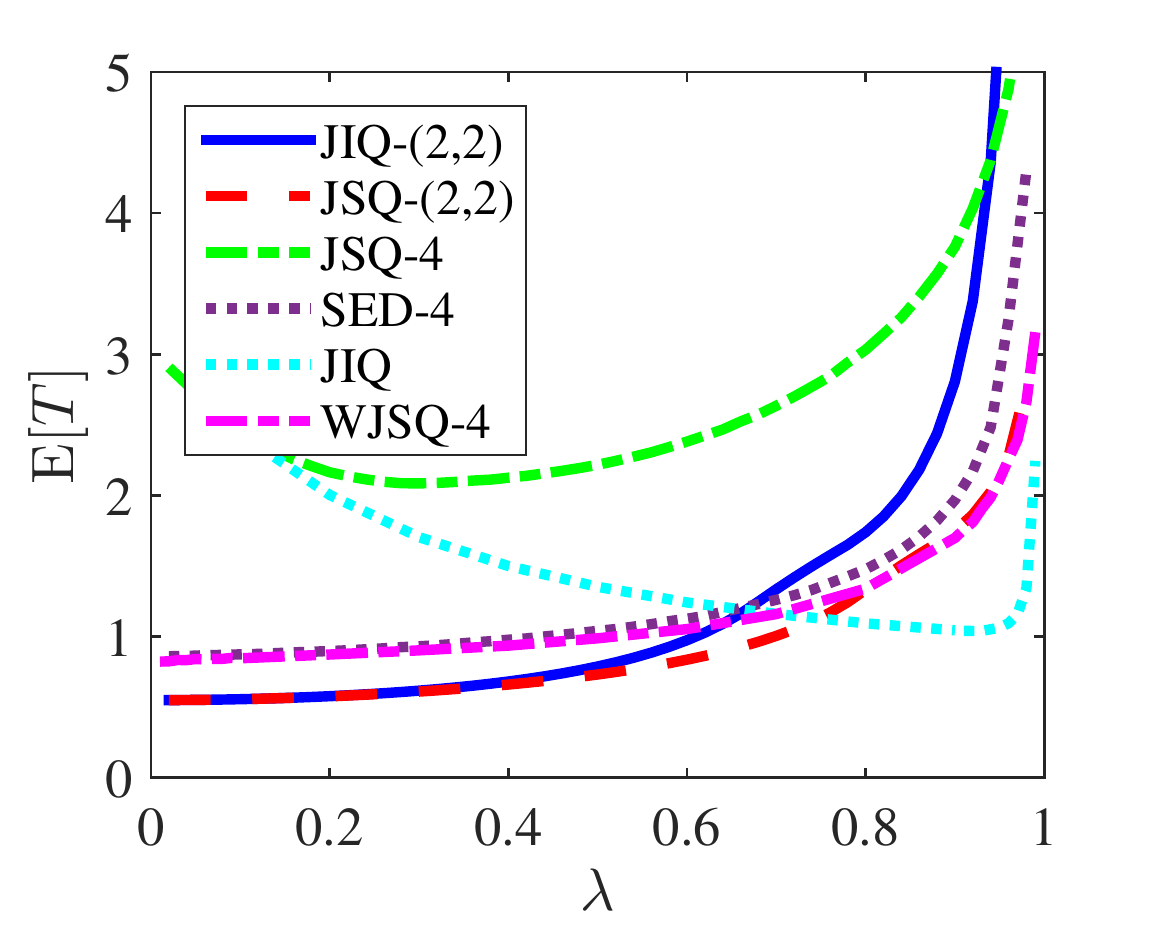} &
			\includegraphics[scale=0.4]{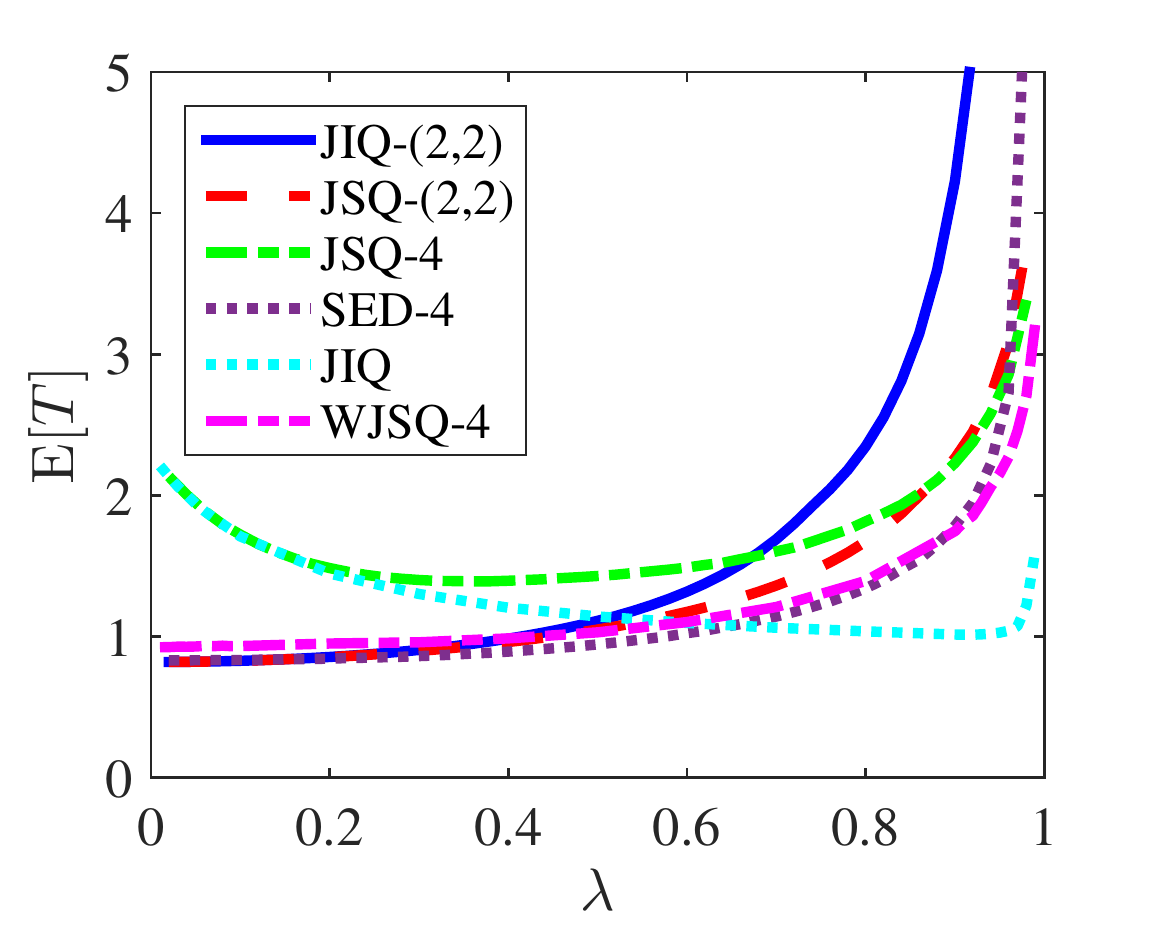} \\
	\end{tabular}}
	\caption{Mean response time as a function of $\lambda$ under JIQ-(2,2), JSQ-(2,2), JSQ-4, SED-4, and JIQ. Left to right: $q_F = 0.2$, $q_F = 0.5$, $q_F = 0.8$. Top to bottom: $r = 1.1$, $r = 2$, $r = 5$, $r = 10$.}
	\label{fig:mean_response_times}
\end{figure}

Figure~\ref{fig:mean_response_times} compares mean response time under JIQ-($d_F$,$d_S$) and JSQ-($d_F$,$d_S$) to that under four other policies (results for our policies are analytical, while results for the following policies are simulated):
\begin{itemize}[leftmargin=*]
	\item Under \textbf{JSQ-$d$}, the dispatcher queries $d$ servers uniformly at random and sends the job to the server among those $d$ with the shortest queue.
	\item Under \textbf{SED-$d$}, the dispatcher queries $d$ servers uniformly at random and sends the job to the server among those $d$ at which it has the shortest expected delay.
	\item Under \textbf{WJSQ-$d$} (the W stands for ``Weighted''), the dispatcher queries $d$ servers, where the probability that a server is queried is proportional to that server's speed, and sends the job to the server among those $d$ with the shortest queue.
	\item Under \textbf{JIQ}, the dispatcher sends the job to an idle server if there is one, and to a busy server chosen uniformly at random otherwise.
\end{itemize}
We note that JSQ-$d$ and JIQ are heterogeneity-unaware, SED-$d$ only uses heterogeneity information when dispatching, and WJSQ-$d$ only uses heterogeneity information when querying. 
Unlike the other five polices that we consider, JIQ is not a ``power of $d$'' policy; we include it here as a point of comparison because it is known to minimize the probability that an arriving job waits in the queue~\cite{stolyar2015pull}.

When there is little difference in speed between fast and slow servers ($r=1.1$, top row of Figure~\ref{fig:mean_response_times}), JSQ-$d$ and SED-$d$ perform similarly to each other, and both outperform our policies at high load.
This is because when all servers are similar in speed, providing more flexibility when selecting among queried servers offers a greater advantage than ensuring that some fast servers are queried.
But in systems with more pronounced heterogeneity, JSQ-$d$ and SED-$d$ cannot maintain their good performance.
As $r$ increases, JSQ-$d$ suffers significantly: here it is a serious shortcoming to make dispatching decisions based only on queue lengths. 
SED-$d$ corrects for this problem by scaling queue lengths in proportion to server speeds.
Yet when $r$ is high and $q_F$ is low, both JSQ-$d$ and SED-$d$ can lead to instability.
In this regime, much of the system's capacity belongs to the fast servers, but an arriving job may not query any fast servers because JSQ-$d$ and SED-$d$ use uniform querying (e.g., when $q_F=0.2$, only about $40\%$ of jobs query a fast server).
This causes the slow servers to become overloaded.
WJSQ-$d$ avoids instability in this regime by ensuring that faster servers are more likely to be queried and thus sent a job. 
However, performance under WJSQ-$d$ still suffers at low load; here all queue lengths are relatively short, so WJSQ-$d$ effectively ignores server speeds when dispatching.

Our policies remain stable and achieve better performance by differentiating between fast and slow servers both when querying and when choosing where to dispatch among the queried servers.
At low load, JIQ-($d_F$,$d_S$) and JSQ-($d_F$,$d_S$) perform similarly to each other, and both outperform SED-$d$, JSQ-$d$, and WJSQ-$d$.
As $r$ increases, the gap between our policies and JSQ-$d$ becomes particularly pronounced: JSQ-$d$ frequently sends jobs to slow servers even when there are idle fast servers, whereas our policies are more likely to find and select an idle fast server.
Indeed, our policies effectively throw out the slow servers when load is sufficiently low or $r$ is sufficiently high.
At high load, too, our policies perform competitively with or better than JSQ-$d$, SED-$d$, and WJSQ-$d$.
Most notably, while JSQ-$d$ and SED-$d$ have a reduced stability region when $q_F$ is low and $r$ is high, both JIQ-($d_F$,$d_S$) and JSQ-($d_F$,$d_S$) are guaranteed to be stable provided $\lambda < \mu_F q_F + \mu_S q_S$, as shown in Theorem~\ref{thm:stability}.

Unsurprisingly, JSQ-($d_F$,$d_S$) always outperforms JIQ-($d_F$,$d_S$).
This makes sense: when using the same $p_F$ and $p_S$ values, the only difference between the two policies is that the JSQ version makes a better dispatching decision when choosing among busy servers.
Note that the results in Figure~\ref{fig:mean_response_times} do \emph{not} necessarily have the same values of $p_F$ and $p_S$ for JSQ-($d_F$,$d_S$) and JIQ-($d_F$,$d_S$) because both policy families are optimized over the parameters. 
Even though JSQ-($d_F$,$d_S$) is guaranteed to achieve lower mean response time than JIQ-($d_F$,$d_S$), the two policies perform similarly until $\lambda$ becomes high.
At this point JSQ-($d_F$,$d_S$)'s advantage becomes more apparent, as this is when queues actually build up.
Under both JIQ-($d_F$,$d_S$) and JSQ-($d_F$,$d_S$), mean response time appears to be non-convex in $\lambda$.
This surprising result is due to our optimization over $p_F$ and $p_S$.
For any fixed $p_F$ and $p_S$, mean response time is convex in $\lambda$, and indeed the convex regions in the plots in Figure~\ref{fig:mean_response_times} occur when $p_F$ and $p_S$ do not change (for example, when $\lambda$ is relatively low it is optimal to set $p_S = 0$, i.e., to never use the slow servers).
The non-convex regions appear when either $p_F$ or $p_S$ is varying between 0 and 1.

We also compare our policies to JIQ, which uses queue length information from all servers, not just a subset of $d$ servers.
At high load, JIQ outperforms all of the ``power of $d$'' policies; this is unsurprising given that JIQ will always find an idle server if there is one.
But at low load and high $r$, JIQ yields a substantially higher mean response time than our policies.
This is because, like JSQ-$d$ and WJSQ-$d$, JIQ does not use server speed information to break ties between idle servers.
That our policies outperform JIQ may seem surprising in light of the fact that JIQ is delay optimal~\cite{stolyar2015pull}; we explore this result further in Section~\ref{sec:response_time_dist}.

\subsection{Queue Length Distribution}
\label{sec:response_time_dist}

\begin{figure}[h]
	\begin{center}
		\begin{tabular}{>{\centering\arraybackslash}m{.6cm}>{\centering\arraybackslash}m{5.0cm}>{\centering\arraybackslash}m{5.0cm}>{\centering\arraybackslash}m{5.0cm}}
			& (a) $r=1.1$, $\lambda=0.5$ & (b) $r=5$, $\lambda=0.8$ & (c) $r=10$, $\lambda=0.2$\\
			\rotatebox{90}{ \textbf{Fast servers}} &
			\includegraphics[scale=0.39]{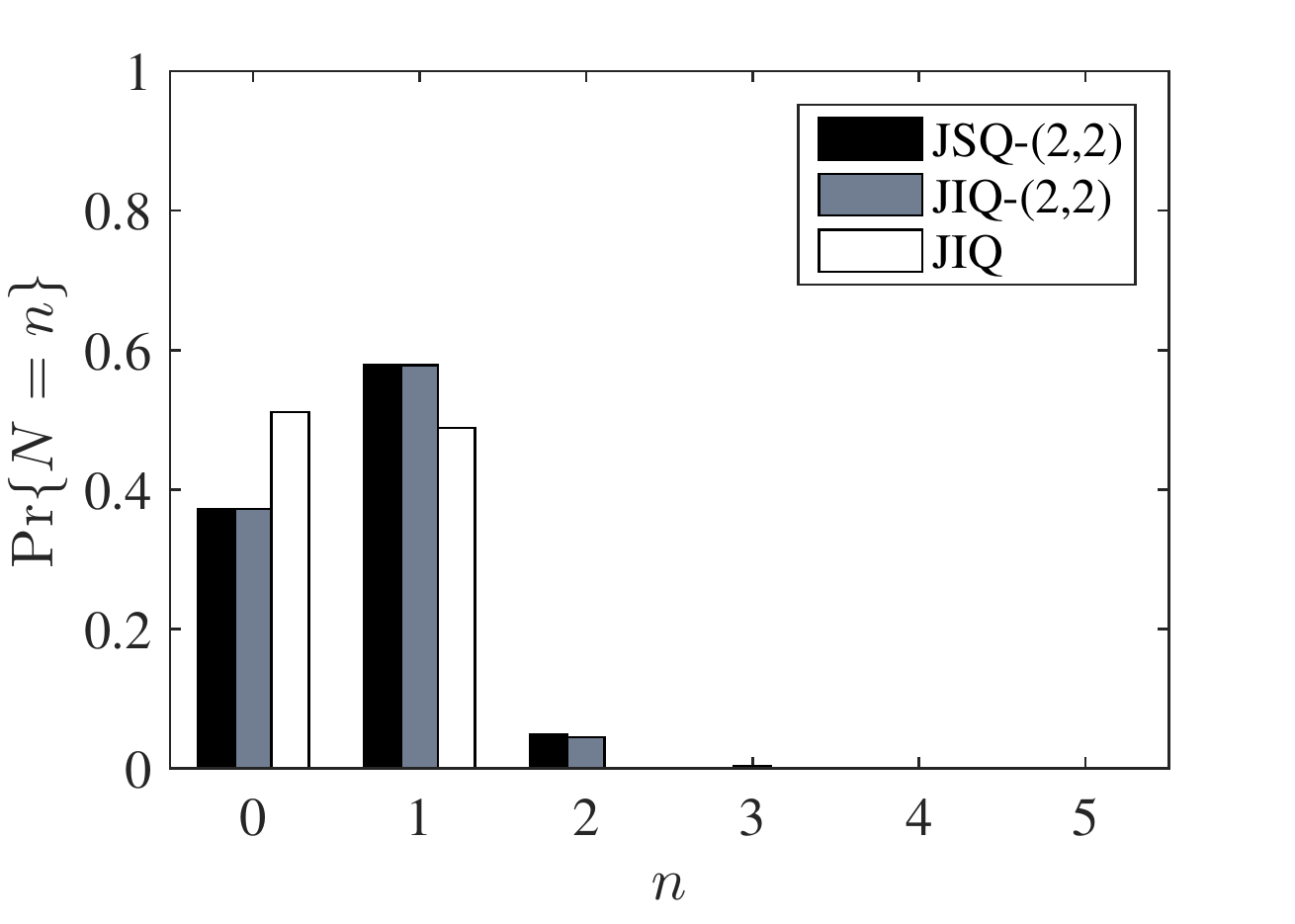} &
			\includegraphics[scale=0.39]{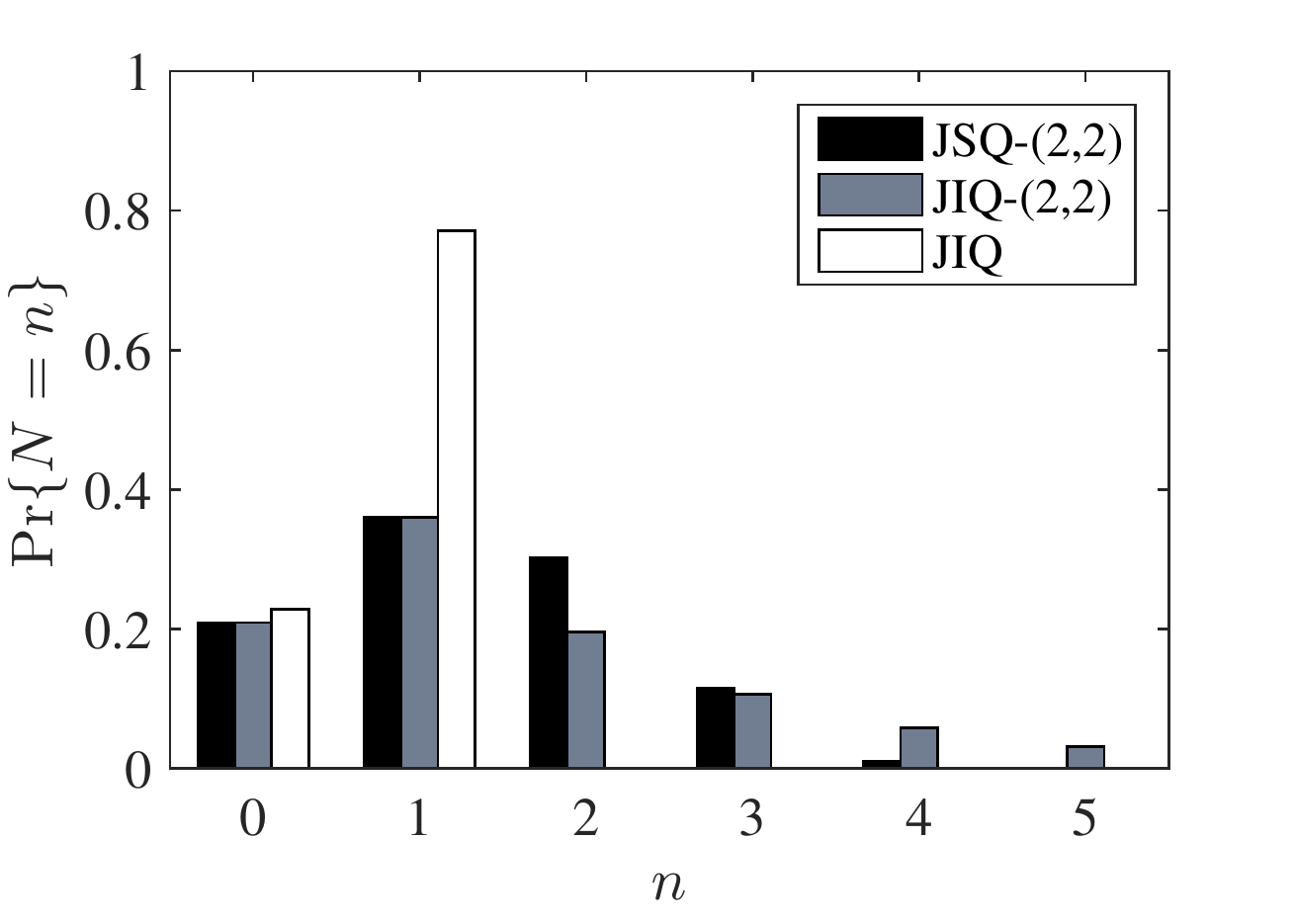} &
			\includegraphics[scale=0.39]{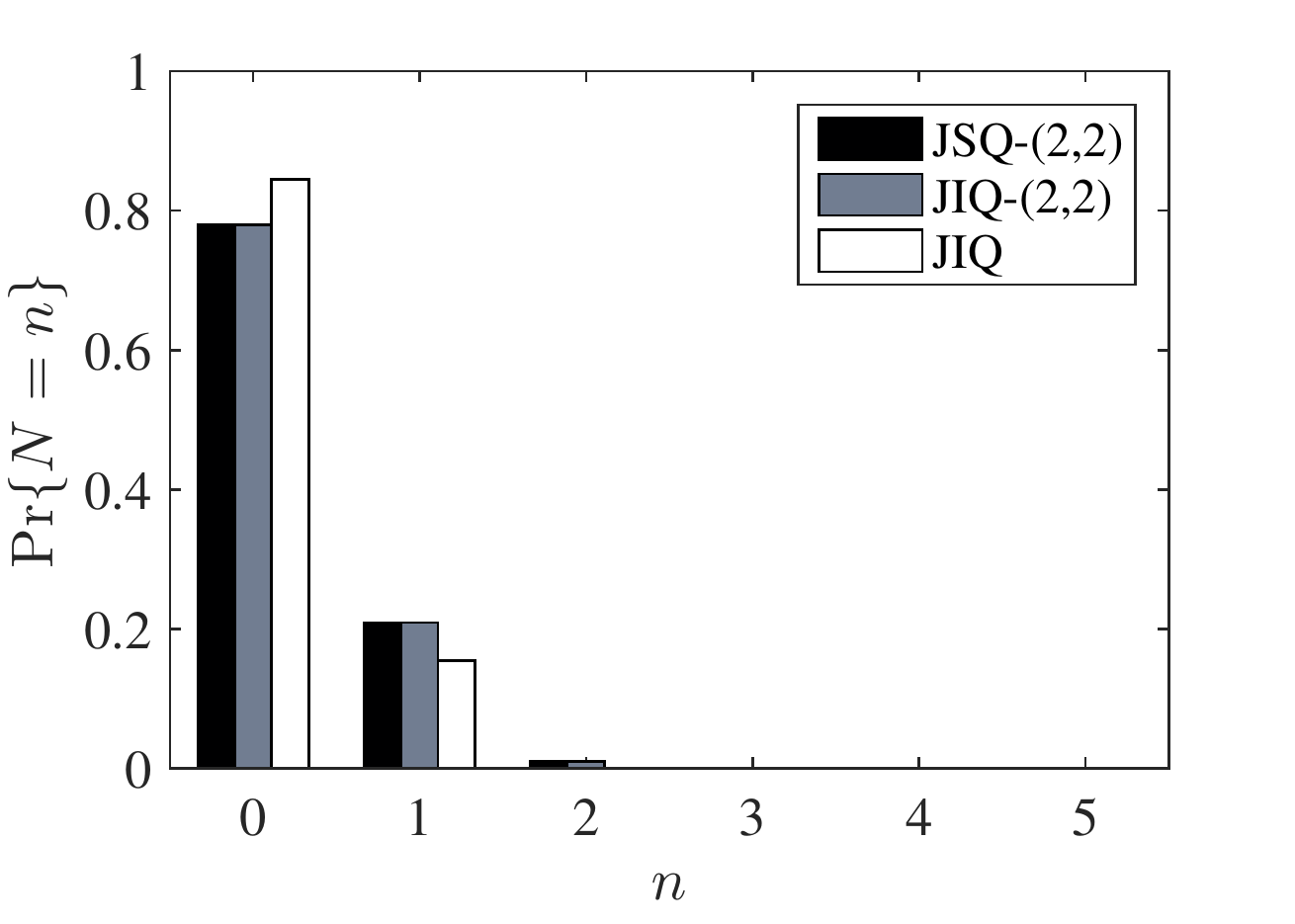} \\
			\rotatebox{90}{ \textbf{Slow servers}} &
			\includegraphics[scale=0.39]{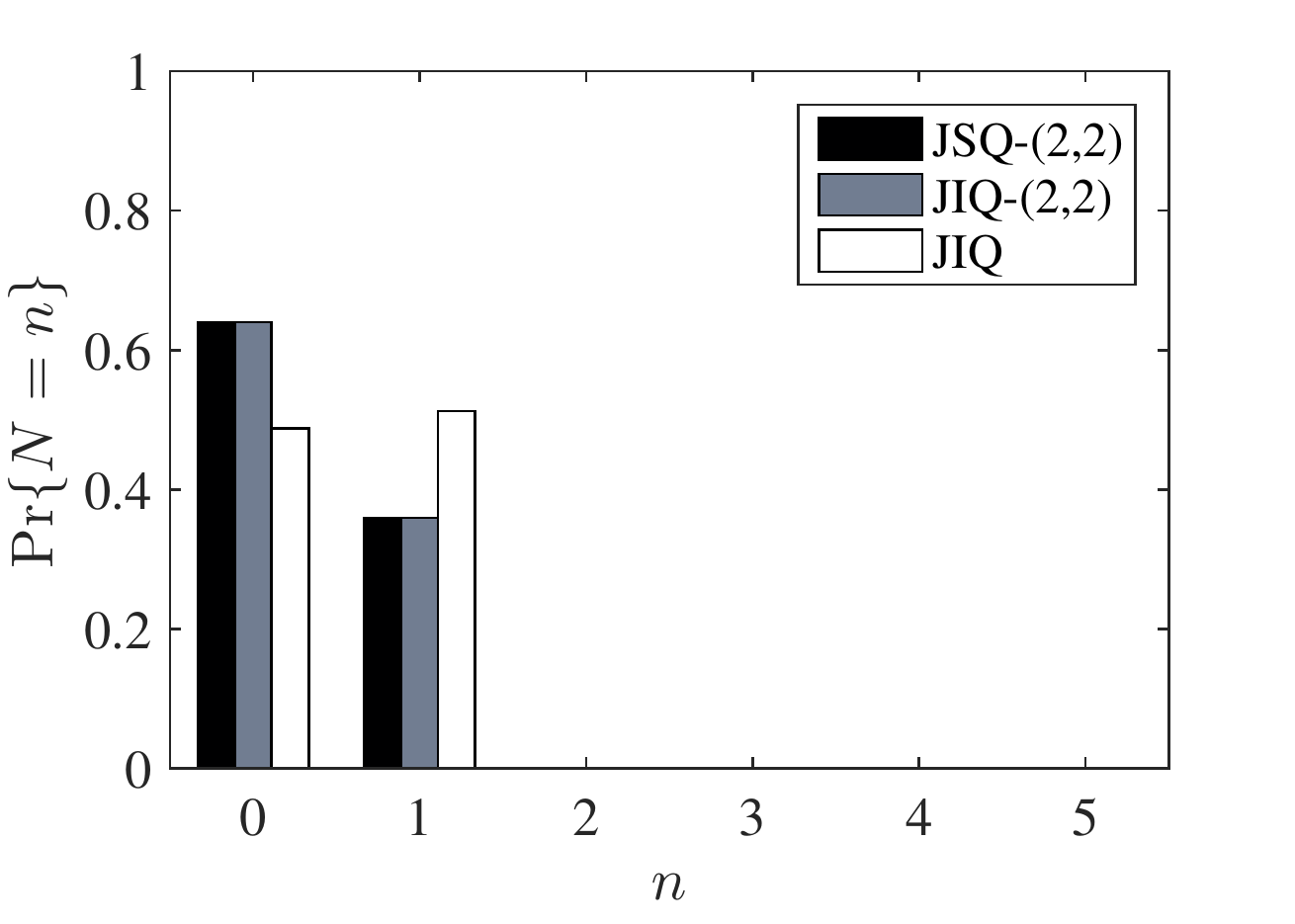} &
			\includegraphics[scale=0.39]{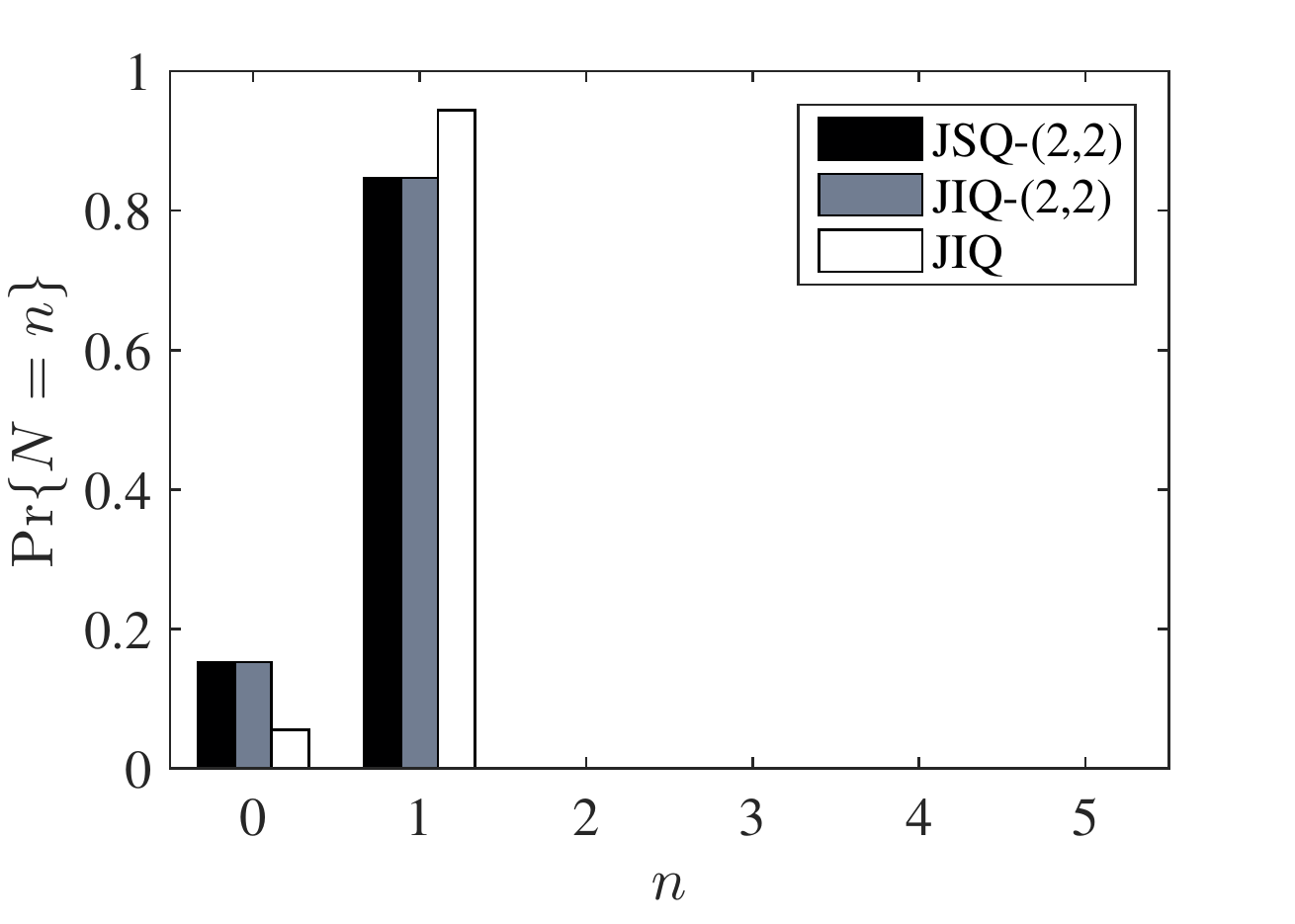} &
			\includegraphics[scale=0.39]{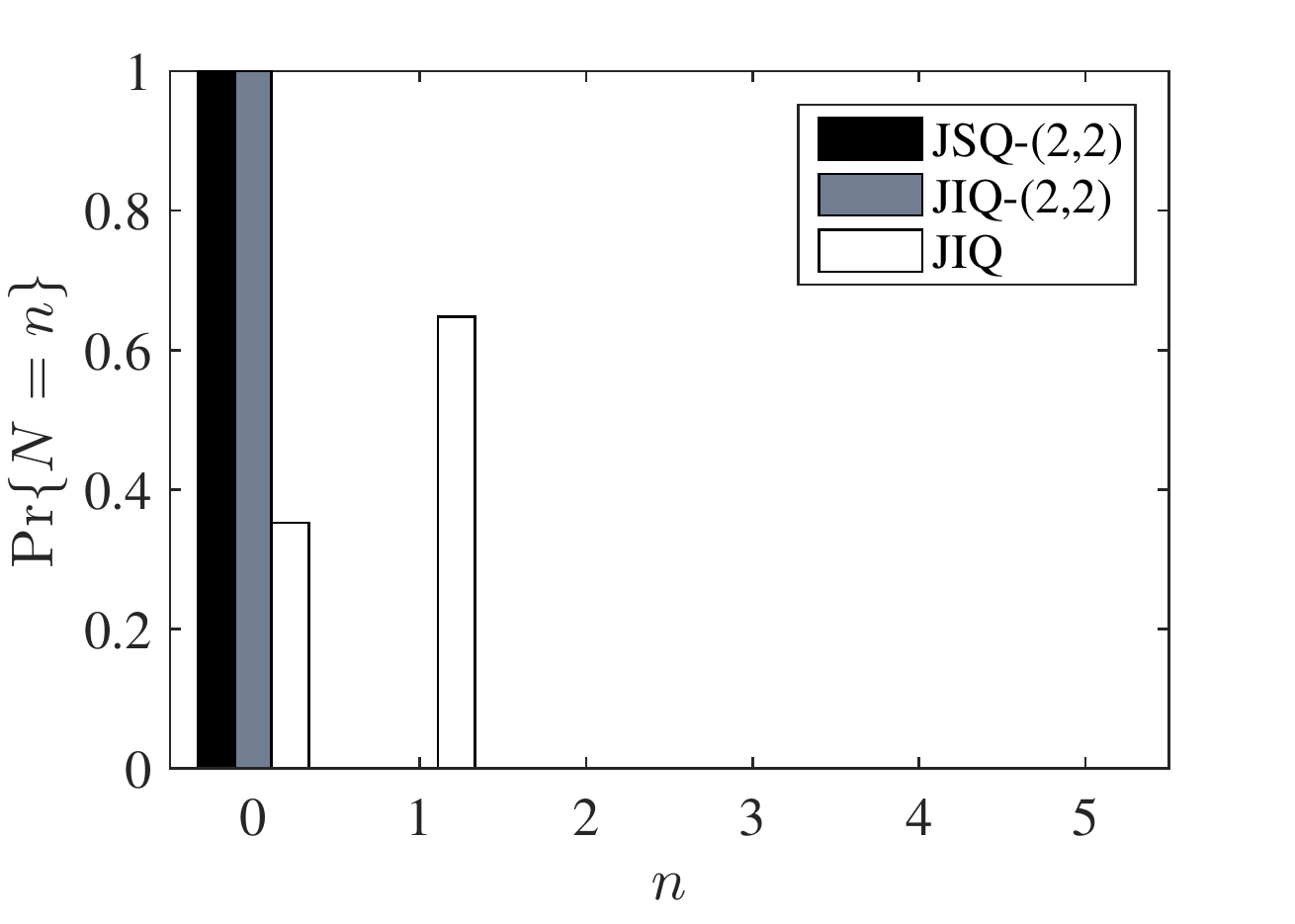}\\
		\end{tabular}
	\end{center}
	\caption{Comparing the queue length distribution under JSQ-(2,2), JIQ-(2,2), and JIQ for fast servers (top row) and slow servers (bottom row) when $q_F = 0.5$. (a) $r=1.1$, $\lambda = 0.5$, (b) $r=5$, $\lambda = 0.8$, (c) $r=10$, $\lambda = 0.2$.}
	\label{fig:queue_length_dists}
\end{figure}

In this section we look at the queue length distributions under JIQ-($d_F,d_S$), JSQ-($d_F,d_S$), and JIQ in more detail to gain insight as to why our policies can outperform JIQ in terms of response time, even though they lack JIQ's queue length optimality property.

Figure~\ref{fig:queue_length_dists} shows the queue length distribution under JIQ-($d_F,d_S$), JSQ-($d_F,d_S$), and JIQ for both fast servers (top row) and slow servers (bottom row) in three settings selected from those featured in Figure~\ref{fig:mean_response_times}.
At left, we show a case in which all three policies have similar mean response times; in this case the queue length distributions are also similar.
The center column shows a case in which JIQ yields lower mean response time than our policies: in this case $r=5$ and $\lambda=0.8$.
Because $\lambda$ is high, few slow servers are idle, but both our policies and JIQ prevent queues from building up at the slow servers.
The key difference between the policies lies in what happens at the fast servers.
Under our policies, the optimal value of $p_F$ in this setting is 1, meaning that a job will never choose to wait in the queue at a slow server.
This means that many jobs are deferred back to the (busy) fast servers, causing the queue lengths to increase.
JIQ prevents the queue lengths at the fast servers from growing.
A slightly greater proportion of jobs run on slow servers under JIQ, but the jobs that run on fast servers do not have to wait in the queue.
When $\lambda$ is high, this tradeoff favors JIQ.

In contrast, when $\lambda$ is low the same tradeoff favors JIQ-($d_F,d_S$) and JSQ-($d_F,d_S$), as shown in the right column of Figure~\ref{fig:queue_length_dists}, where $r=10$ and $\lambda=0.2$.
Again, under JIQ a higher proportion of slow servers are busy because JIQ does not differentiate between fast and slow idle servers.
Indeed, there are no busy slow servers under JIQ-($d_F,d_S$) and JSQ-($d_F,d_S$) because the combination of high $r$ and low $\lambda$ means that the optimal value of $p_S$ is 0: it is best not to use any of the slow servers at all.
As a result, the fast servers have a slightly lower probability of being idle under our policies than under JIQ.
However, because $\lambda$ is low the queue lengths under JIQ-($d_F,d_S$) and JSQ-($d_F,d_S$) remain short.
In this case, JIQ's decision to prioritize server idleness over server speed works against it, and our policies achieve lower mean response time.

\subsection{Sensitivity to $d$}

One of the primary selling points of policies like JSQ-$d$, SED-$d$, and WJSQ-$d$ is the ``power of two choices'': often, there is a large benefit in going from $d=1$ (i.e., random routing) to $d=2$, but a much smaller marginal benefit in further increasing $d$.
Consequently, JSQ-2 is the most commonly considered variant of JSQ-$d$.
Our JIQ-($d_F$,$d_S$) and JSQ-($d_F$,$d_S$) policies query fast and slow servers separately; while setting $d_F = d_S = 1$ offers two choices in total, it does not offer a choice within each speed.
Therefore, JIQ-(1,1) and JSQ-(1,1) are equivalent: once the dispatcher has chosen to send the job to a fast (or slow) server there is only one choice for which server to use.
Henceforth, we will refer to both policies as JIQ-(1,1).

Unlike JSQ-2 and SED-2, JIQ-(1,1) uses queue length information only when deciding between an idle slow server and a busy fast server; all other decisions are made probabilistically.  This makes JIQ-(1,1) much closer to random routing than either JSQ-2 or SED-2, and one might think that consequently JIQ-(1,1) would generally exhibit poor performance.  However, our results indicate the opposite: JIQ-(1,1) often substantially outperforms JSQ-2 and SED-2, especially when $q_F$ is low (see Figure~\ref{fig:power_of_2}).
As we have seen, both JSQ-2 and SED-2 can cause instability when $q_F$ is low and $r$ is high, whereas JIQ-(1,1) guarantees that the system will remain stable.

\begin{figure}
	\begin{center}
	\begin{tabular}{cc}
		(a) $q_F = 0.2$ & (b) $q_F = 0.8$ \\
		\hspace{-0.1in} \includegraphics[scale=0.36]{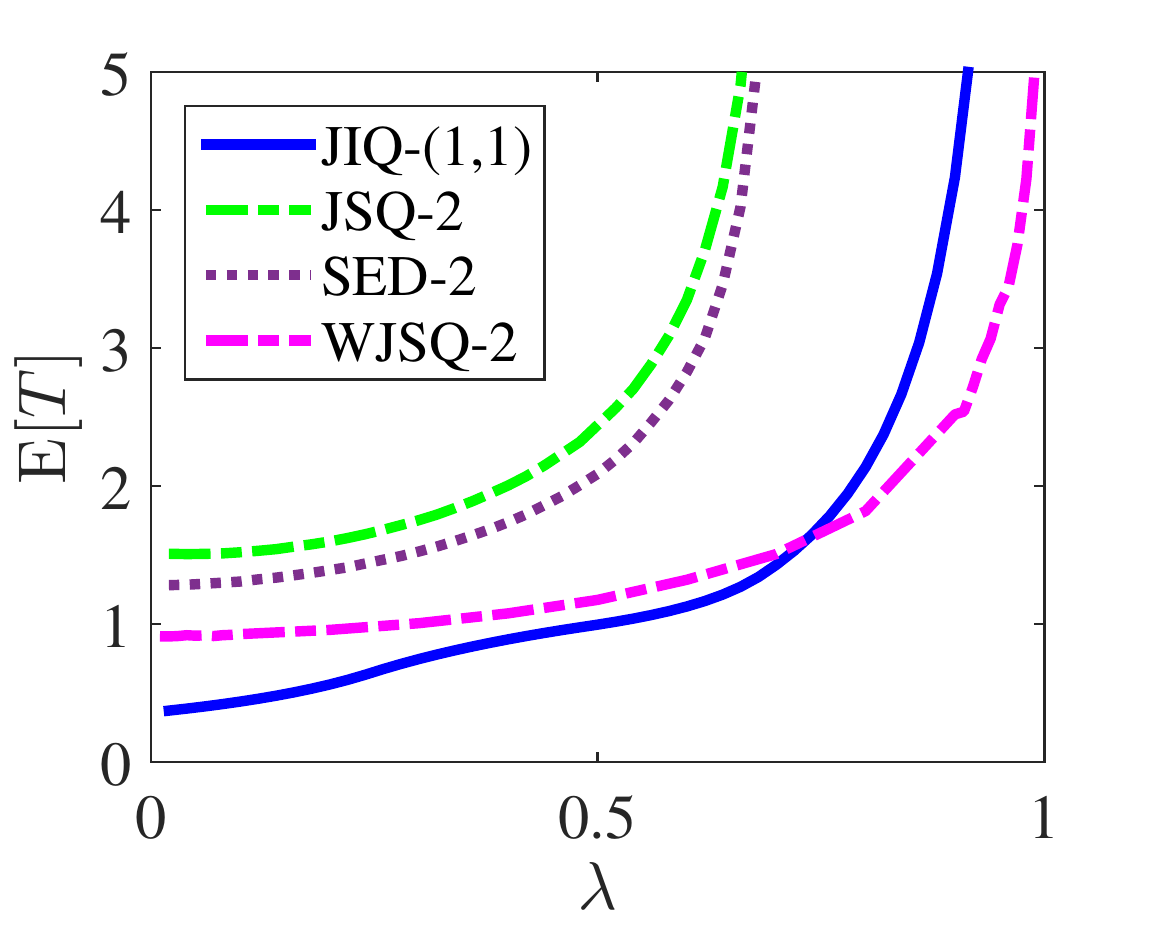} &
		\includegraphics[scale=0.36]{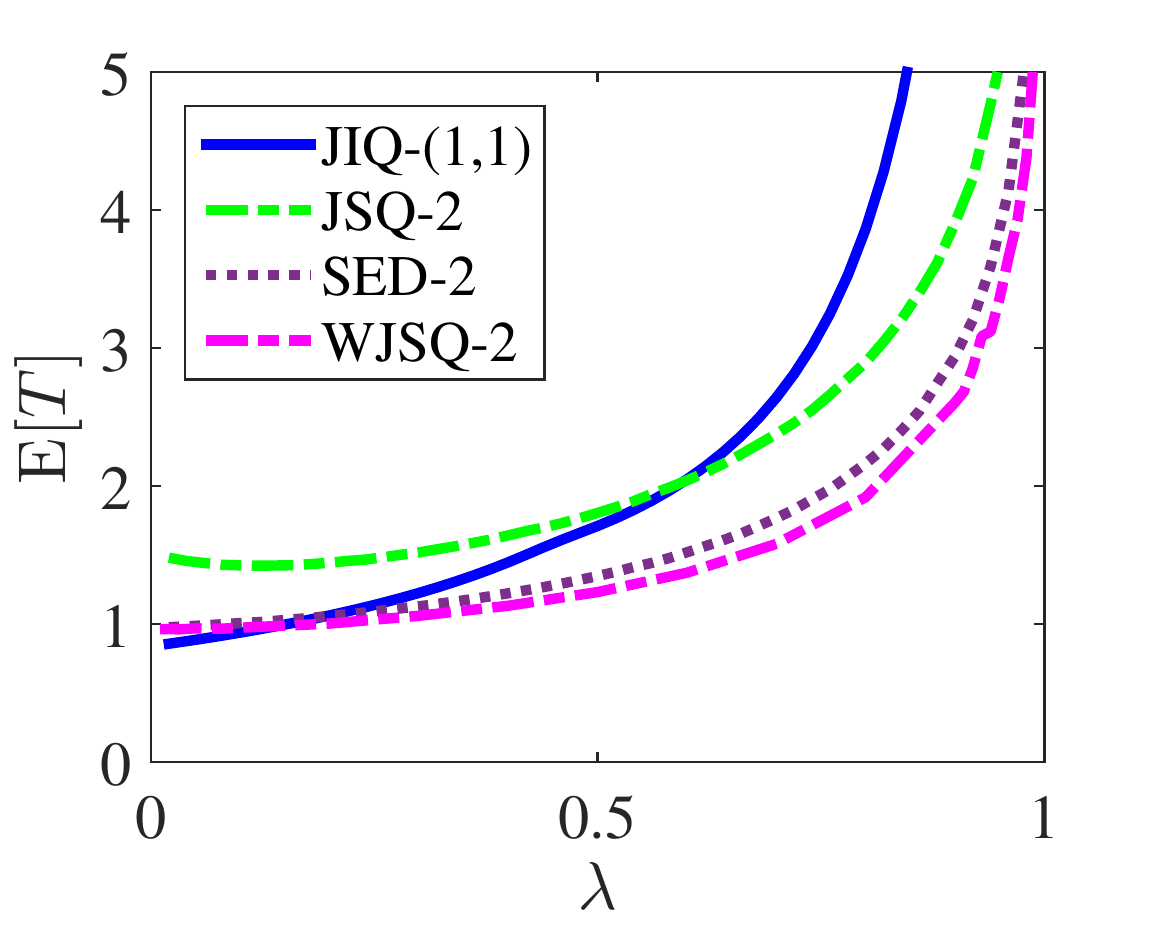} 
	\end{tabular}
\end{center}
	\caption{Mean response time as a function of $\lambda$ under JIQ-(1,1), JSQ-2, SED-2, and WJSQ-2 when $r=5$. (a) $q_F = 0.2$, (b)~$q_F = 0.8$.} 
	\label{fig:power_of_2}
\end{figure}

In Figure~\ref{fig:varyd} we consider the effect of varying $d = d_F + d_S$ on the performance of JIQ-($d_F$,$d_S$) and JSQ-($d_F$,$d_S$): does the marginal benefit of increasing $d$ decrease as $d$ gets larger?
When $d=1$, we interpret our policies to collapse the querying and dispatching decision points into a single probabilistic choice: we dispatch to a random fast server with probability $p_F$ and to a slow server otherwise.
For all other values of $d$, we choose the optimal combination of $d_F$, $d_S$, $p_F$, and $p_S$ such that $d_F + d_S = d$.
As under JSQ-$d$ and SED-$d$, the steepest drop in mean response time comes from going from $d=1$ to $d=2$, and mean response time is convex in $d$.
When the fast and slow servers are similar in speed (Figure~\ref{fig:varyd}~(a)), JSQ-$d$ and SED-$d$ perform slightly better at low $d$, and all policies have similar performance at high $d$.
When the $r$ is high and $q_F$ is low (Figure~\ref{fig:varyd}~(b)), JIQ-($d_F$,$d_S$) and JSQ-($d_F$,$d_S$) are stable at all values of $d$, and outperform JSQ-$d$ and SED-$d$ even when $d$ is high enough for the latter two policies to be stable.

\begin{figure}
	\begin{center}
	\begin{tabular}{>{\centering\arraybackslash\vspace{.2cm}} m{6.5cm} >{\centering\arraybackslash\vspace{-.2cm}} m{6cm}}
		\multicolumn{2}{c}{(a) $q_F=0.5$, $r=1.1$}\\
		\includegraphics[scale=0.45]{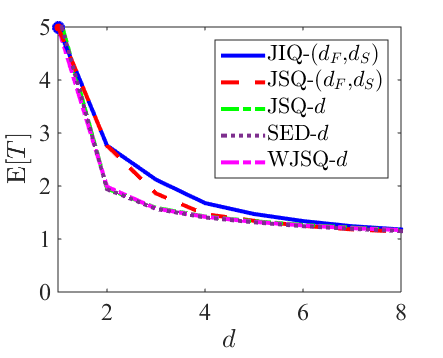} & \small{\begin{tabular}{|c|c|c|}
				\hline
				$d$ & JIQ-($d_F$,$d_S$) & JSQ-($d_F$,$d_S$) \\
				\hline
				2 & (1,1) & (1,1)\\
				3 & (1,2) & (2,1)\\
				4 & (2,2) & (2,2)\\
				5 & (2,3) & (2,3)\\
				6 & (3,3) & (3,3)\\
				7 & (3,4) & (3,4)\\
				8 & (3,5) & (3,5)\\
				\hline
		\end{tabular}}\\
		\hline
		\multicolumn{2}{c}{(b) $q_F=0.2$, $r=10$}\\
		\includegraphics[scale=0.45]{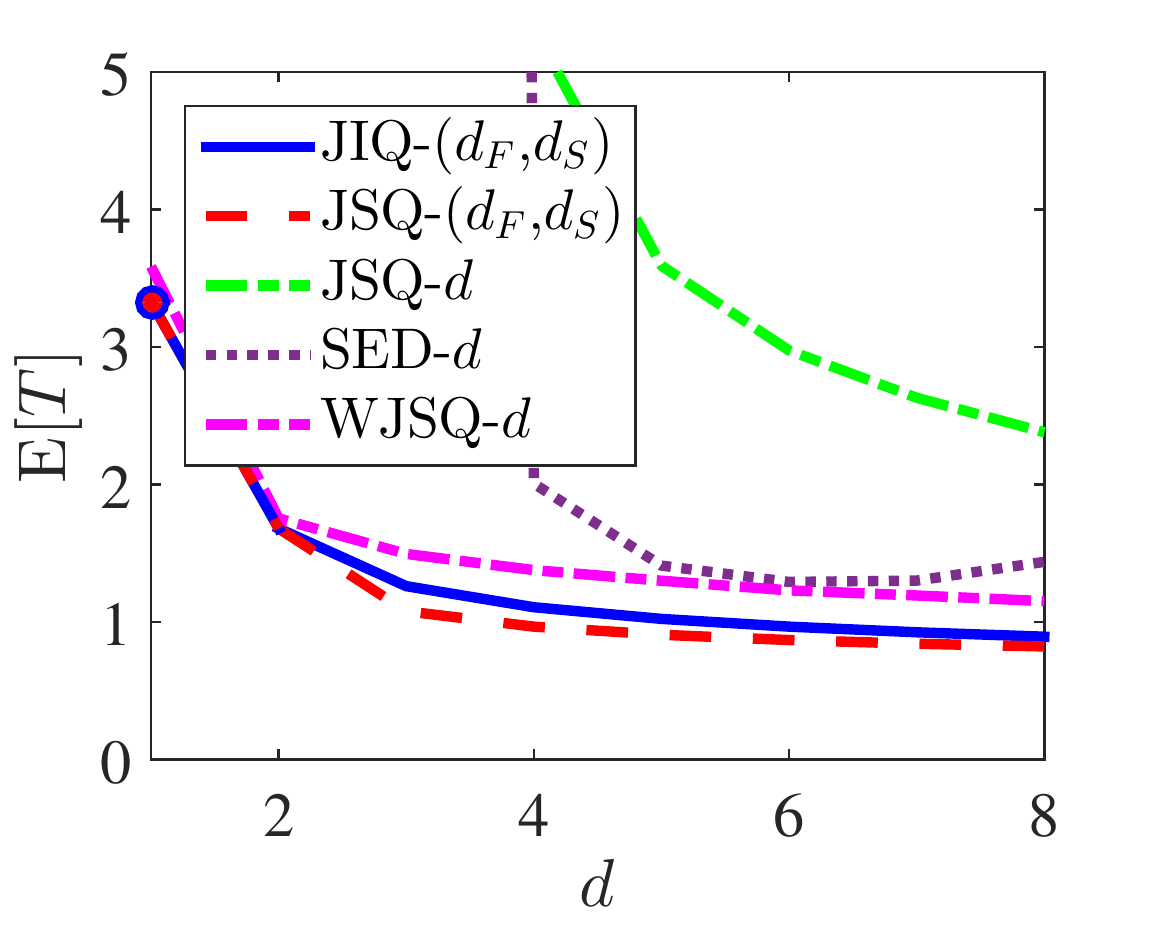} & \small{\begin{tabular}{|c|c|c|}
				\hline
				$d$ & JIQ-($d_F$,$d_S$) & JSQ-($d_F$,$d_S$) \\
				\hline
				2 & (1,1) & (1,1)\\
				3 & (2,1) & (2,1)\\
				4 & (3,1) & (3,1)\\
				5 & (4,1) & (4,1)\\
				6 & (5,1) & (5,1)\\
				7 & (6,1) & (6,1)\\
				8 & (7,1) & (7,1)\\
				\hline
		\end{tabular}}\\
	\end{tabular}
\end{center}
	\caption{Effect of varying $d$ on mean response time under JIQ-($d_F$,$d_S$), JSQ-($d_F$,$d_S$), JSQ-$d$, SED-$d$, and WJSQ-$d$ when $\lambda=0.8$. (a) $q_F=0.5$, $r=1.1$. (b) $q_F=0.2$, $r=10$. The tables at right show the optimal choices of ($d_F,d_S$) for each $d$.}
	\label{fig:varyd}
\end{figure}

\section{A Heuristic for $p_F$ and $p_S$}
\label{sec:heuristic}

A key part of defining the JIQ-($d_F$,$d_S$) and JSQ-($d_F$,$d_S$) policies involves choosing values for $p_F$ and $p_S$; in Sections~\ref{sec:jiq-opt} and~\ref{sec:jsq-opt} we do this by finding the values of $p_F$ and $p_S$ that minimize mean response time.
Figure~\ref{fig:optimization} shows mean response time under JSQ-($d_F$,$d_S$) as a function of $p_F$ and $p_S$ for two different parameter settings (results for JIQ-($d_F$,$d_S$) are similar).
When $\lambda$ is low to moderate (Figure~\ref{fig:optimization}(a)), mean response time is relatively insensitive to the particular parameter choices, provided that $p_S$ is high enough to ensure stability.
When $\lambda$ is high (Figure~\ref{fig:optimization}(b)), it becomes more important to choose the right $p_F$ and $p_S$: even small variations in $p_F$ and $p_S$ can lead to substantial changes in response time, and there is a smaller set of $p_F$ and $p_S$ values for which the system is stable.

\begin{figure}
	\begin{center}
	\begin{tabular}{cc}
		(a) $q_F=0.2$, $r=5$, $\lambda = 0.56$ & (b) $q_F=0.5$, $r=2$, $\lambda=0.95$\\
		\includegraphics[scale=0.4]{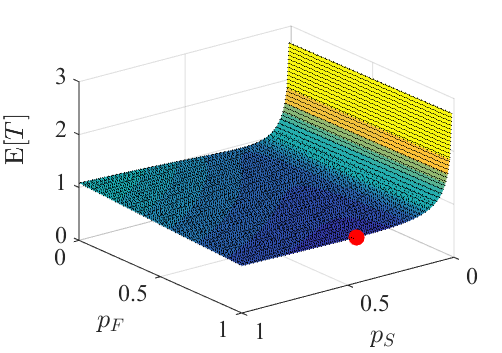} &		
		\includegraphics[scale=0.4]{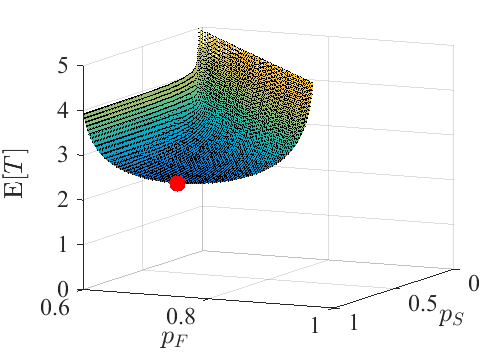}
	\end{tabular}
	\end{center}	
	\caption{Mean response time as a function of $p_F$ and $p_S$. (a)~ $q_F=0.2$, $r=5$, $\lambda = 0.56$, (b) $q_F=0.5$, $r=2$, $\lambda=0.95$. The red circle indicates the optimal $\E{T}$.}
	\label{fig:optimization}
\end{figure}

The extreme sensitivity to $p_F$ and $p_S$ occurs only at very high $\lambda$; at most parameter settings the optimal values of $p_F$ and $p_S$ fall into one of a few cases.
If the fast servers comprise a sufficiently high fraction of the total system capacity or if the system load is very low, it is best to set $p_S = 0$.
If the fast and slow servers are relatively similar in speed or if the system load is sufficiently high, it is best to set $p_S = 1$.
As we showed in Theorem~\ref{thm:stability_highlambda}, as $\lambda \rightarrow 1$, $p_F = \mu_F q_F$ is the only value of $p_F$ for which the system is stable.

Motivated by these observations, we propose a heuristic for choosing appropriate values of $p_F$ and $p_S$.
Instead of optimizing over the entire parameter space for $p_F$ and $p_S$, which can be computationally expensive, we consider the following parameter settings:
\begin{itemize}[leftmargin=*]
	\item $p_S = 0$. Note that in this case the slow servers are never used, so the choice of $p_F$ does not matter.
	\item All combinations of $p_S \in \{ \mu_S q_S, 1 \}$ and $p_F \in \{0, \mu_F q_F, 1 \}$.
\end{itemize}
For each setting of $\lambda$, $q_F$, and $r$, this gives us only seven policies to compare; we select the $p_F$ and $p_S$ that yields the best performance among these seven alternatives.

Table~\ref{tab:heuristics_jiq} shows our results for JIQ-($d_F$,$d_S$) and JSQ-($d_F$,$d_S$); each row shows a different value of $\lambda$, for a system with $q_F=0.2$ and $r=10$. 
Under both policies, when $\lambda$ is low it is optimal to set $p_S = 0$, and our heuristic correctly selects this policy.
As $\lambda$ starts to increase, it becomes optimal to increase $p_S$ continuously and set $p_F = 1$.
Our heuristic sets $p_F = 1$ and changes $p_S$ in discrete steps from 0 to $\mu_S q_S$ to 1; because $\lambda$ is still relatively low, mean response time is relatively insensitive to selecting a slightly suboptimal value of $p_S$ and our heuristic has low error.
When $\lambda$ becomes high, the performance of our heuristic can suffer. 
In this region it becomes optimal to set $p_S = 1$ and decrease $p_F$ continuously, while our heuristic must choose either $p_F = 1$ or $p_F = \mu_F q_F$.
Because $\lambda$ is high, a small change in $p_F$ (which corresponds to a small change in the arrival rate to any individual busy server), can have a big affect on mean response time, and the error of our heuristic can reach as high as 25\%.
However, as $\lambda \rightarrow 1$, the heuristic, which sets $p_S = 1$ and $p_F = \mu_F q_F$, again approaches perfect accuracy because $p_F = \mu_F q_F$ is the only value of $p_F$ that maintains stability, and as $\lambda \rightarrow 1$ the queue lengths build up so using an idle slow server when one is available ($p_S = 1$) also should be optimal.

\begin{table*}[t]
	\centering
	\begin{tabular}{|l|lllllll|}
		\hline
		& \multicolumn{7}{c||}{JIQ-(2,2)} \\ 
		$\lambda$ & $p_F^{*}$ & $p_S^{*}$ & $\E{T_{\mathrm{opt}}}$ & $p_F^{\mathrm{heur}}$ & $p_S^{\mathrm{heur}}$ & $\E{T_{\mathrm{heur}}}$ & $\%$ error \\ \hline 
		0.14 & any & 0 & 0.384 & any & 0 & 0.384 & 0 \\  \hline 
		0.24 & any & 0 & 0.443 & any & 0 & 0.443 & 0 \\ \hline 
		0.34 & 0.999 & 0.018 & 0.575 & any & 0 & 0.576 & 0.023 \\ \hline 
		0.44 & 1 & 0.426 & 0.742 & 1 & 0.444 & 0.743 & 0.014 \\ \hline 
		0.54 & 1 & 0.723 & 0.868 & 1 & 1 & 0.879 & 1.196 \\ \hline 
		0.64 & 1 & 1 & 0.967 & 1 & 1 & 0.967 & 0 \\ \hline
		0.74 & 1 & 1 & 1.101 & 1 & 1 & 1.101 & 0 \\ \hline
		0.84 & 0.877 & 1 & 1.547 & 1 & 1 & 1.605 & 3.732 \\ \hline 
		0.90 & 0.714 & 1 & 2.331 & 0.555 & 1 & 2.908 & 24.754 \\ \hline 
		0.98 & 0.579 & 1 & 10.677 & 0.555 & 1 & 12.837 & 20.231 \\ \hline 
		& \multicolumn{7}{c|}{JSQ-(2,2)} \\
		$\lambda$ & $p_F^{*}$ & $p_S^{*}$ & $\E{T_{\mathrm{opt}}}$ & $p_F^{\mathrm{heur}}$ & $p_S^{\mathrm{heur}}$ & $\E{T_{\mathrm{heur}}}$ & $\%$ error \\ \hline 
		0.14 & any & 0 & 0.383 & any & 0 & 0.383 & 0 \\  \hline 
		0.24 & any & 0 & 0.429 & any & 0 & 0.429 & 0 \\ \hline 
		0.34 & any & 0 & 0.514 & any & 0 & 0.514 & 0 \\ \hline 
		0.44 & 1 & 0.103 & 0.677 & any & 0 & 0.689 & 1.693 \\ \hline 
		0.54 & 1 & 0.405 & 0.832 & 1 & 0.444 & 0.833 & 0.066 \\ \hline 
		0.64 & 1 & 0.722 & 0.946 & 1 & 1 & 0.954 & 0.762 \\ \hline
		0.74 & 1 & 1 & 1.039 & 1 & 1 & 1.039 & 0 \\ \hline
		0.84 & 1 & 1 & 1.217 & 1 & 1 & 1.217 & 0 \\ \hline 
		0.90 & 0.839 & 1 & 1.595 & 1 & 1 & 1.957 & 22.697 \\ \hline 
		0.98 & 0.597 & 1 & 3.243 & 0.555 & 1 & 3.659 & 12.804 \\ \hline 
	\end{tabular}
	\caption{Comparison of optimal $p_F$ and $p_S$ to best heuristic under JIQ-(2,2) (top) and JSQ-(2,2) (bottom). Here $q_F = 0.2$ and $r = 5$. The columns $p_F^{*}$ and $p_S^{*}$ give the optimal values of $p_F$ and $p_S$, while $p_F^{\mathrm{heur}}$ and $p_S^{\mathrm{heur}}$ are the values chosen by the heuristic.}
	\label{tab:heuristics_jiq}
\end{table*}

\section{Conclusion}
\label{sec:conclusion}

This paper addresses the problem of dispatching in large-scale, heterogeneous systems.
We design two new heterogeneity-aware families of policies, JIQ-($d_F$,$d_S$) and JSQ-($d_F$,$d_S$).
Our policies are simple, analytically tractable, and provide outstanding performance.

Our results yield several insights about how to design ``power of $d$'' policies that perform well in heterogeneous settings.
In order to maintain the maximum stability region, the dispatcher must ensure that fast servers are queried sufficiently often.
Alone, neither uniform sampling nor weighting querying in favor of fast servers is enough to ensure good performance.
Our work establishes that, instead, dispatching policies should use heterogeneity information at two decision points: (1) when choosing which servers to query, and (2) when choosing where among the queried servers to dispatch a job.
Ultimately, how best to distribute jobs among fast and slow servers depends jointly on the system load, the fraction of servers that are fast, and the relative speeds of the servers.
It may be best to use only fast servers, to use slow servers only when they are idle, or to balance jobs among fast and slow servers in some other way.
Because there is no single right answer, policies designed for heterogeneous systems must be able to adapt to the system parameters.
JIQ-($d_F$,$d_S$) and JSQ-($d_F$,$d_S$) do this by optimizing over the probabilistic parameters to choose the best allocation of jobs to fast and slow servers.
Moreover, as we show in Theorem~\ref{thm:stability}, the optimal policy in each family is guaranteed to be stable.

We focus specifically on policies that query fixed numbers of fast and slow servers and then make probabilistic decisions about how to route among the queried servers based on idleness and queue length information.  
The space of policies that use heterogeneity information at both decision points is much larger than the policies we propose here.
For example, one could imagine generalizing our policies at the first decision point by choosing $d_F$ and $d_S$ probabilistically for each query; this also would allow us to adapt our policies for systems with more than two server speeds.
At the second decision point, one could combine ($d_F,d_S$)-style querying with a heterogeneity-aware dispatching policy, such as SED.
While optimizing over such a large policy space is likely to be challenging, we are optimistic that substantial advances could be made in future work toward understanding a wider scope of policies and settings.

Differing server speeds is just one way in which server farms may exhibit heterogeneity.
Systems may also consist of servers that are heterogeneous in their memory, network bandwidth, or any other resource availability.
Some jobs may be able to run on certain servers but not on others, for example due to data locality.
Jobs may be capable of running on any server, but may have a preference for or run faster on certain servers.
The policies we present in this paper are designed to perform well specifically for the case of heterogeneous server speeds, but we believe the insights gained will aid the design of effective load balancing policies for the broad range of heterogeneity that exists in today's systems.

\bibliographystyle{abbrv}
\bibliography{references}

\begin{thebibliography}{10}

\bibitem{banawan1992comparative}
S.~Banawan and N.~Zeidat.
\newblock A comparative study of load sharing in heterogeneous multicomputer
  systems.
\newblock In {\em Proceedings. 25th Annual Simulation Symposium}, pages 22--31.
  IEEE, 1992.

\bibitem{banawan1989load}
S.~A. Banawan and J.~Zahorjan.
\newblock Load sharing in heterogeneous queueing systems.
\newblock In {\em In Proc. of {IEEE INFOCOM'89}}, pages 731--739, 1989.

\bibitem{bonomi-tc-1990}
F.~Bonomi.
\newblock On job assignment for a parallel system of processor sharing queues.
\newblock {\em IEEE Trans. Comput.}, 39(7):858--869, July 1990.

\bibitem{chen2012asymptotic}
H.~Chen and H.-Q. Ye.
\newblock Asymptotic optimality of balanced routing.
\newblock {\em Operations research}, 60(1):163--179, 2012.

\bibitem{FENG2005}
H.~Feng, V.~Misra, and D.~Rubenstein.
\newblock Optimal state-free, size-aware dispatching for heterogeneous
  m/g/-type systems.
\newblock {\em Performance Evaluation}, 62(1):475 -- 492, 2005.
\newblock Performance 2005.

\bibitem{gupta2007analysis}
V.~Gupta, M.~Harchol-Balter, K.~Sigman, and W.~Whitt.
\newblock Analysis of join-the-shortest-queue routing for web server farms.
\newblock {\em Performance Evaluation}, 64(9-12):1062--1081, 2007.

\bibitem{Mor-2013}
M.~Harchol-Balter.
\newblock {\em Performance Modeling and Design of Computer Systems: Queueing
  Theory in Action}.
\newblock Cambridge University Press, 2013.

\bibitem{esa2013}
E.~Hyytiä.
\newblock Optimal routing of fixed size jobs to two parallel servers.
\newblock {\em INFOR: Information Systems and Operational Research},
  51(4):215--224, 2013.

\bibitem{izagirre2014light}
A.~Izagirre and A.~Makowski.
\newblock Light traffic performance under the power of two load balancing
  strategy: the case of server heterogeneity.
\newblock {\em SIGMETRICS Performance Evaluation Review}, 42(2):18--20, 2014.

\bibitem{koole-scl-1995}
G.~Koole.
\newblock A simple proof of the optimality of a threshold policy in a
  two-server queueing system.
\newblock {\em Systems and Control Letters}, 26(5):301--303, Dec. 1995.

\bibitem{larsen81}
R.~L. Larsen.
\newblock {\em {Control of Multiple Exponential Servers with Application to
  Computer Systems}}.
\newblock PhD thesis, College Park, MD, USA, 1981.

\bibitem{lin84}
W.~Lin and P.~R. Kumar.
\newblock {Optimal Control of a Queueing System with Two Heterogeneous
  Servers}.
\newblock {\em IEEE Transactions on Automatic Control}, 29(8):696--703, 1984.

\bibitem{lu2011join}
Y.~Lu, Q.~Xie, G.~Kliot, A.~Geller, J.~Larus, and A.~Greenberg.
\newblock Join-idle-queue: A novel load balancing algorithm for dynamically
  scalable web services.
\newblock {\em Performance Evaluation}, 68(11):1056--1071, 2011.

\bibitem{luh2002}
H.~P. Luh and I.~Viniotis.
\newblock Threshold control policies for heterogeneous server systems.
\newblock {\em Mathematical Methods of Operations Research}, 55(1):121--142,
  2002.

\bibitem{mitzenmacher2001power}
M.~Mitzenmacher.
\newblock The power of two choices in randomized load balancing.
\newblock {\em IEEE Transactions on Parallel and Distributed Systems},
  12(10):1094--1104, 2001.

\bibitem{mukhopadhyay2016analysis}
A.~Mukhopadhyay and R.~Mazumdar.
\newblock Analysis of randomized join-the-shortest-queue (jsq) schemes in large
  heterogeneous processor-sharing systems.
\newblock {\em IEEE Transactions on Control of Network Systems}, 3(2):116--126,
  2016.

\bibitem{nelson1989approximation}
R.~D. Nelson and T.~K. Philips.
\newblock {\em An approximation to the response time for shortest queue
  routing}, volume~17.
\newblock ACM, 1989.

\bibitem{rubinovitch85}
M.~Rubinovitch.
\newblock {The Slow Server Problem}.
\newblock {\em Journal of Applied Probability}, 22(1):205--213, 1985.

\bibitem{rubinovitch85_stall}
M.~Rubinovitch.
\newblock {The Slow Server Problem: A Queue with Stalling}.
\newblock {\em Journal of Applied Probability}, 22(4):879--892, 1985.

\bibitem{rykov09}
V.~V. Rykov and D.~V. Efrosinin.
\newblock On the slow server problem.
\newblock {\em Automation and Remote Control}, 70(12):2013--2023, 2009.

\bibitem{selen2016approximate}
J.~Selen, I.~Adan, and S.~Kapodistria.
\newblock Approximate performance analysis of generalized join the shortest
  queue routing.
\newblock In {\em Proceedings of the 9th EAI International Conference on
  Performance Evaluation Methodologies and Tools}, pages 103--110. ICST
  (Institute for Computer Sciences, Social-Informatics and~…, 2016.

\bibitem{selen2016steady}
J.~Selen, I.~Adan, S.~Kapodistria, and J.~van Leeuwaarden.
\newblock Steady-state analysis of shortest expected delay routing.
\newblock {\em Queueing Systems}, 84(3-4):309--354, 2016.

\bibitem{Sethuraman:1999}
J.~Sethuraman and M.~S. Squillante.
\newblock Optimal stochastic scheduling in multiclass parallel queues.
\newblock {\em SIGMETRICS Perform. Eval. Rev.}, 27(1):93--102, May 1999.

\bibitem{shenker1989}
S.~{Shenker} and A.~{Weinrib}.
\newblock The optimal control of heterogeneous queueing systems: a paradigm for
  load-sharing and routing.
\newblock {\em IEEE Transactions on Computers}, 38(12):1724--1735, Dec 1989.

\bibitem{stolyar2015pull}
A.~Stolyar.
\newblock Pull-based load distribution in large-scale heterogeneous service
  systems.
\newblock {\em Queueing Systems}, 80(4):341--361, 2015.

\bibitem{tantawi1985optimal}
A.~N. Tantawi and D.~Towsley.
\newblock Optimal static load balancing in distributed computer systems.
\newblock {\em Journal of the ACM (JACM)}, 32(2):445--465, 1985.

\bibitem{vvedenskaya1996queueing}
N.~Vvedenskaya, R.~Dobrushin, and F.~Karpelevich.
\newblock Queueing system with selection of the shortest of two queues: An
  asymptotic approach.
\newblock {\em Problemy Peredachi Informatsii}, 32(1):20--34, 1996.

\bibitem{wang2018distributed}
C.~Wang, C.~Feng, and J.~Cheng.
\newblock Distributed join-the-idle-queue for low latency cloud services.
\newblock {\em IEEE/ACM Transactions on Networking}, 26(5):2309--2319, 2018.

\bibitem{weber1978optimal}
R.~R. Weber.
\newblock On the optimal assignment of customers to parallel servers.
\newblock {\em Journal of Applied Probability}, 15(2):406--413, 1978.

\bibitem{whitt1986deciding}
W.~Whitt.
\newblock Deciding which queue to join: Some counterexamples.
\newblock {\em Operations research}, 34(1):55--62, 1986.

\bibitem{winston1977optimality}
W.~Winston.
\newblock Optimality of the shortest line discipline.
\newblock {\em Journal of Applied Probability}, 14(1):181--189, 1977.

\bibitem{zhou2017designing}
X.~Zhou, F.~Wu, J.~Tan, Y.~Sun, and N.~Shroff.
\newblock Designing low-complexity heavy-traffic delay-optimal load balancing
  schemes: Theory to algorithms.
\newblock {\em Proceedings of the ACM on Measurement and Analysis of Computing
  Systems}, 1(2):39, 2017.

\end{thebibliography}

\newpage

\section*{Appendix}
\label{sec:appendix}

Here we give the complete expanded form of the optimization formulations given in (\ref{opt:jiq}, \ref{opt:jsq}).  

For JIQ-($d_F, d_S$) our optimization formulation (\ref{opt:jiq}) is as follows:
\begin{align*}
& \underset{p_F, p_S}{\text{minimize}} 
& &\ \ \ \frac{q_F \lambda_{IF}\mu_F}{\lambda (\mu_F - \lambda_{BF})(\mu_F - \lambda_{BF} + \lambda_{IF})}\\
& & &+ \frac{q_S \lambda_{IS}\mu_S}{\lambda (\mu_S - \lambda_{BS})(\mu_S - \lambda_{BS} + \lambda_{IS})}\\
& \text{subject to} 
& & \lambda_{IF} = \frac{\lambda d_F}{q_F} \left( \sum_{i=0}^{d_F-1} \binom{d_F-1}{i} \frac{\pi_{0F}^{i} (1 - \pi_{0F})^{d_F-1-i}}{i+1}  \right)\\
& & & \lambda_{BF} = \frac{\lambda d_F}{q_F} \frac{(1-\pi_{0F})^{d_F-1}}{d_F} \\
& & &\quad\cdot \left( \left( 1 - (1 - \pi_{0S})^{d_S} \right) (1-p_S) + (1 - \pi_{0S})^{d_S} p_F  \right)\\
& & & \lambda_{IS} = \frac{\lambda d_S}{q_S} (1-\pi_{0F})^{d_F} \\
& & &\quad\cdot \left( \sum_{i=0}^{d_S-1} \binom{d_S - 1}{i} \frac{\pi_{0S}^{i}(1 - \pi_{0S})^{d_S-1-i}}{i+1} p_S \right)\\
& & & \lambda_{BS} = \frac{\lambda d_S}{q_S} (1-\pi_{0F})^{d_F} \frac{(1 - \pi_{0S})^{d_S-1}}{d_S} (1-p_F)\\
& & & \pi_{0F} = \frac{\mu_F - \lambda_{BF}}{\mu_F - \lambda_{BF} + \lambda_{IF}}\\
& & & \pi_{0S} = \frac{\mu_S - \lambda_{BS}}{\mu_S - \lambda_{BS} + \lambda_{IS}}\\
& & & 0 \leq \pi_{0F}, \pi_{0S} \leq 1 \\
& & & 0 \leq p_F, p_S \leq 1
\end{align*}

For JSQ-($d_F, d_S$) our optimization formulation (\ref{opt:jsq}) is as follows:
\begin{align*}
& \underset{p_F, p_S}{\text{minimize}} 
& & \frac{1}{\mu_F} \cdot \left(1-\rho_F^{d_F}\right) + \frac{1}{\mu_S} \cdot \rho_F^{d_F}\left(1-\rho_S^{d_s}\right)p_S \\
& & &\quad + \frac{1}{\mu_F} \sum_{i=1}^{\infty} \left(i+1\right) \cdot \frac{f_i^{d_F} - f_{i+1}^{d_F}}{f_1^{d_F}} \\ 
& & & \quad\quad\cdot\rho_F^{d_F}(\rho_S^{d_S} p_F + \left(1-\rho_S^{d_s}\right)\left(1-p_S)\right) \\
& & &\quad + \frac{1}{\mu_S} \sum_{i=1}^{\infty} \left(i+1\right) \cdot \frac{s_i^{d_S} - s_{i+1}^{d_S}}{s_1^{d_S}} \cdot \rho_F^{d_F} \rho_S^{d_S}\left(1-p_F\right)\\
& \text{subject to} 
& & \rho_F = \frac{\lambda}{\mu_F q_F} \left( \rho_F^{d_F}\left(1- \rho_S^{d_S}\right)\left(1-p_S\right) \right)\\
& & &\quad\quad+ \frac{\lambda}{\mu_F q_F}\left( \left(1-\rho_F^{d_F}\right) + \rho_F^{d_F} \rho_S^{d_S} p_F \right)\\
& & & \rho_S = \frac{\lambda}{\mu_S q_S} \left( \rho_F^{d_F} \left(1-\rho_S^{d_S}\right)p_S + \rho_F^{d_F} \rho_S^{d_S}\left(1-p_F\right) \right)\\
& & & \frac{\lambda}{q_F} \left(f_{i-1}^{d_F} - f_i^{d_F}\right) \left( \left(1 - \rho_S^{d_S}\right)\left(1 - p_S\right) + \rho_S^{d_S} p_F \right) \\
& & &\quad\quad= \mu_F \left(f_i - f_{i+1}\right)  \qquad \qquad \qquad \qquad \qquad \ \, i \ge 1\\
& & & \frac{\lambda}{q_S} \left(s_{i-1}^{d_S} - s_i^{d_S}\right) \rho_F^{d_F} \left(1 - p_F\right) = \mu_S \left(s_i - s_{i+1}\right) \quad i \ge 1\\
& & & f_0 = s_0 = 1 \\
& & & f_1 = \rho_F \\
& & & s_1 = \rho_S \\
& & & 0 \leq \rho_F, \rho_S \leq 1 \\
& & & 0 \leq p_F, p_S \leq 1
\end{align*}

\end{document}